%% file: main.tex
\def\llncs{0}
\def\fullpage{1}
\def\anonymous{0}
\def\authnote{0}
\def\notxfont{0}
\def\submission{0}

\ifnum\submission=1
\def\llncs{1}
\fi

\ifnum\llncs=1 	\documentclass[envcountsect,a4paper,runningheads,10pt]{llncs}
\else
	\documentclass[letterpaper,hmargin=1.0in,vmargin=1.0in,11pt]{article}
			\ifnum\fullpage=1
		\usepackage{fullpage}
		\fi
\fi

\input{preamble_usepackages.tex}
\input{macros}
 
\title{Quantum Cryptography and Hardness of Non-Collapsing Measurements}

\ifnum\anonymous=1
\ifnum\llncs=1
\author{\empty}\institute{\empty}
\else
\author{}
\fi
\else
\ifnum\llncs=1
\author{
Tomoyuki Morimae\inst{1} \and Yuki Shirakawa\inst{1} \and Takashi Yamakawa\inst{2,3,1}
}
\institute{
 Yukawa Institute for Theoretical Physics, Kyoto University, Kyoto, Japan \and NTT Social Informatics Laboratories, Tokyo, Japan \and NTT Research Center for Theoretical Quantum Information, Atsugi, Japan
}
\else
\author[1]{Tomoyuki Morimae}
\author[1]{ Yuki Shirakawa}
\author[2,3,1]{ Takashi Yamakawa}
\affil[1]{{\small Yukawa Institute for Theoretical Physics, Kyoto University, Kyoto, Japan}\authorcr{\small tomoyuki.morimae@yukawa.kyoto-u.ac.jp} \authorcr{\small yuki.shirakawa@yukawa.kyoto-u.ac.jp}}
\affil[2]{{\small NTT Social Informatics Laboratories, Tokyo, Japan}\authorcr{\small takashi.yamakawa@ntt.com}}
\affil[3]{{\small NTT Research Center for Theoretical Quantum Information, Atsugi, Japan}}
\fi 
\fi

\date{}

\begin{document}

\maketitle

\begin{abstract}
One-way puzzles (OWPuzzs) introduced by Khurana and Tomer [STOC 2024] are a natural quantum analogue of one-way functions (OWFs),
and one of the most fundamental primitives in ``Microcrypt'' where OWFs do not exist but quantum cryptography is possible.
OWPuzzs are implied by almost all quantum cryptographic primitives, and imply several important
applications such as non-interactive commitments and multi-party computations.
A significant goal in the field of quantum cryptography is to base OWPuzzs on plausible assumptions that will not imply OWFs.
In this paper, we base OWPuzzs on hardness of non-collapsing measurements.
To that end, we introduce a new complexity class, $\mathbf{SampPDQP}$,
which is a sampling version of the decision class $\mathbf{PDQP}$
introduced in [Aaronson, Bouland, Fitzsimons, and Lee, ITCS 2016].
We show that 
if $\mathbf{SampPDQP}$ is hard on average for quantum polynomial time,
then OWPuzzs exist. We also show that
if $\mathbf{SampPDQP}\not\subseteq\mathbf{SampBQP}$, then auxiliary-input OWPuzzs exist.
$\mathbf{SampPDQP}$ is 
the class of sampling problems that can be solved by a classical polynomial-time deterministic
algorithm that can make a single query to a non-collapsing measurement oracle,
which is a ``magical'' oracle that can sample measurement results on quantum states
without collapsing the states.
Such non-collapsing measurements are highly unphysical operations
that should be hard to realize
in quantum polynomial-time,
and therefore our assumptions on which OWPuzzs are based seem extremely plausible.
Moreover, our assumptions do not seem to imply OWFs,
because the possibility of inverting classical functions would not
be helpful to realize quantum non-collapsing measurements.\mor{koko kaeru?}
We also study upperbounds of the hardness of $\mathbf{SampPDQP}$.
We introduce a new primitive, {\it distributional collision-resistant puzzles (dCRPuzzs)},
which are a natural quantum analogue of distributional collision-resistant hashing [Dubrov and Ishai, STOC 2006].
We show that dCRPuzzs imply average-case hardness of $\mathbf{SampPDQP}$ (and therefore OWPuzzs as well).
We also show that two-message honest-statistically-hiding commitments with classical communication
and one-shot message authentication codes (MACs), which are a privately-verifiable version of one-shot signatures [Amos, Georgiou, Kiayias, Zhandry, STOC 2020],
imply dCRPuzzs.
\end{abstract}

\newpage

 \setcounter{tocdepth}{2}
 \tableofcontents
 \newpage

\input{introduction}

\input{preliminaries}

\input{sampling}
\input{PDQP_OWPuzz_average}

\input{adaptivePDQP}
\input{CROWPuzz}

\input{oneshot}

\input{commitments}

\ifnum\anonymous=1
\else
{\bf Acknowledgements.}
TM is supported by
JST CREST JPMJCR23I3,
JST Moonshot R\verb|&|D JPMJMS2061-5-1-1, 
JST FOREST, 
MEXT QLEAP, 
the Grant-in Aid for Transformative Research Areas (A) 21H05183,
and 
the Grant-in-Aid for Scientific Research (A) No.22H00522.
YS is supported by JST SPRING, Grant Number JPMJSP2110.
\fi

\ifnum\submission=0
\bibliographystyle{alpha} 
\else
\bibliographystyle{splncs04}
\fi
\bibliography{abbrev3,crypto,reference,text}


\end{document}

%% file: preamble_usepackages.tex
\usepackage[%
  colorlinks=true,
  citecolor=blue,
  pagebackref=true
]{hyperref}

\usepackage{amsmath, amsfonts, amssymb, mathtools,amscd}

\usepackage{amsthm}

\usepackage{lmodern}
\usepackage[T1]{fontenc}
\usepackage[utf8]{inputenc}

\usepackage{arydshln} 
\usepackage{url}
\usepackage{ifthen}
\usepackage{bm}
\usepackage{multirow}
\usepackage[dvips]{graphicx}
\usepackage[usenames]{color}
\usepackage{xcolor,colortbl} 
\usepackage{threeparttable}
\usepackage{comment}
\usepackage{paralist,verbatim}
\usepackage{cases}
\usepackage{booktabs}
\usepackage{braket}
\usepackage{cancel} 
\usepackage{ascmac} 
\usepackage{framed}
\usepackage{authblk}
\usepackage{pifont}
\usepackage{qcircuit}
\usepackage{tikz}
\definecolor{darkblue}{rgb}{0,0,0.6}
\definecolor{darkgreen}{rgb}{0,0.5,0}
\definecolor{maroon}{rgb}{0.5,0.1,0.1}
\definecolor{dpurple}{rgb}{0.2,0,0.65}

\usepackage[capitalise,noabbrev]{cleveref}
\usepackage[absolute]{textpos}
\usepackage[final]{microtype}
\usepackage[absolute]{textpos}
\usepackage{everypage}
\DeclareMathAlphabet{\mathpzc}{OT1}{pzc}{m}{it}

\usepackage{algorithmic}
\usepackage{algorithm}
\usepackage{here}

\newtheoremstyle{thicktheorem}%
{\topsep}
{\topsep}
{\itshape}{}%
{\bfseries}%
{.}
{ }%
{\thmname{#1}\thmnumber{ #2}%
		\thmnote{ (#3)}%
}

\newtheoremstyle{remark}
{\topsep}
{\topsep}
	{}
	{}
	{}
	{.}
	{ }
	{\textit{\thmname{#1}}\thmnumber{ #2}
			\thmnote{ (#3)}%
	}

\ifnum\llncs=0
	\theoremstyle{thicktheorem}
	\newtheorem{theorem}{Theorem}[section]
	\newtheorem{lemma}[theorem]{Lemma}
	\newtheorem{corollary}[theorem]{Corollary}
	
	\newtheorem{definition}[theorem]{Definition}

	\theoremstyle{remark}
	
	\newtheorem{remark}[theorem]{Remark}

\else
\fi

\Crefname{MyClaim}{Claim}{Claims}

	\crefname{theorem}{Theorem}{Theorems}
	\crefname{assumption}{Assumption}{Assumptions}
	\crefname{construction}{Construction}{Constructions}
	\crefname{corollary}{Corollary}{Corollaries}
	\crefname{conjecture}{Conjecture}{Conjectures}
	\crefname{definition}{Definition}{Definitions}
	\crefname{example}{Example}{Examples}
	\crefname{experiment}{Experiment}{Experiments}
	\crefname{counterexample}{Counterexample}{Counterexamples}
	\crefname{lemma}{Lemma}{Lemmata}
	\crefname{observation}{Observation}{Observations}
	\crefname{proposition}{Proposition}{Propositions}
	\crefname{remark}{Remark}{Remarks}
	\crefname{claim}{Claim}{Claims}
	\crefname{fact}{Fact}{Facts}
	\crefname{note}{Note}{Notes}

\ifnum\llncs=1
 \crefname{appendix}{App.}{Appendices}
 \crefname{section}{Sec.}{Sections}
\else
\fi

\ifnum\llncs=1
\renewcommand*{\backref}[1]{}
\else
	\renewcommand*{\backref}[1]{(Cited on page~#1.)}
	\ifnum\notxfont=1
	\else
		\usepackage{newtxtext}
	\fi
\fi

%% file: macros.tex
\ifnum\authnote=0  
\newcommand{\mor}[1]{}
\newcommand{\minki}[1]{}
\newcommand{\takashi}[1]{}

\else
\newcommand{\mor}[1]{$\ll$\textsf{\color{red} Tomoyuki: { #1}}$\gg$}
\newcommand{\takashi}[1]{$\ll$\textsf{\color{orange} Takashi: { #1}}$\gg$}

\fi

\newcommand{\Tr}{\mathrm{Tr}}
\newcommand{\com}{\mathsf{com}}
\newcommand{\SD}{\mathsf{SD}} 
\newcommand{\sigk}{\mathsf{sigk}} 
\newcommand{\mvk}{\mathsf{mvk}} 

\newcommand{\sR}{\mathsf{R}}
\newcommand{\sS}{\mathsf{S}}

\newcommand{\puzz}{\mathsf{puzz}}
\newcommand{\ans}{\mathsf{ans}}

\newcommand{\Samp}{\algo{Samp}}

\newcommand{\Bad}{\mathtt{Bad}}

\newcommand{\floor}[1]{\left\lfloor{#1}\right\rfloor}

\newcommand{\cA}{\mathcal{A}}
\newcommand{\cB}{\mathcal{B}}
\newcommand{\cC}{\mathcal{C}}
\newcommand{\cD}{\mathcal{D}}
\newcommand{\cE}{\mathcal{E}}
\newcommand{\cF}{\mathcal{F}}

\newcommand{\cH}{\mathcal{H}}

\newcommand{\cQ}{\mathcal{Q}}
\newcommand{\cR}{\mathcal{R}}

\def\makeuppercase#1{
\expandafter\newcommand\csname tl#1\endcsname{\widetilde{#1}}
}

\def\makelowercase#1{
\expandafter\newcommand\csname tl#1\endcsname{\widetilde{#1}}
}

\newcommand{\N}{\mathbb{N}}

\newcommand{\regC}{\mathbf{C}}

\newcommand{\regD}{\mathbf{D}}

\newcommand{\regB}{\mathbf{B}}
\newcommand{\regA}{\mathbf{A}}
\newcommand{\regE}{\mathbf{E}}

\newcommand{\secp}{\lambda}

\newcommand*{\msk}{\keys{msk}}

\newcommand*{\pp}{\keys{pp}}

\newcommand*{\vk}{\keys{vk}}

\newcommand*{\keys}[1]{\mathsf{#1}}

\newcommand*{\algo}[1]{\ensuremath{\mathsf{#1}}}

\newenvironment{boxfig}[2]{\begin{figure}[#1]\fbox{\begin{minipage}{0.97\linewidth}
                        \vspace{0.2em}
                        \makebox[0.025\linewidth]{}
                        \begin{minipage}{0.95\linewidth}
            {{
                        #2 }}
                        \end{minipage}
                        \vspace{0.2em}
                        \end{minipage}}}{\end{figure}}

\newcommand{\bit}{\{0,1\}}

\newcommand{\Setup}{\algo{Setup}}

\newcommand{\Gen}{\algo{Gen}}

\newcommand{\Sign}{\algo{Sign}}

\newcommand{\Ver}{\algo{Ver}}

\newcommand{\TD}{\algo{TD}}

\newcommand{\negl}{{\mathsf{negl}}}

\newcommand{\poly}{{\mathrm{poly}}}

\makeatletter
\DeclareRobustCommand
  \myvdots{\vbox{\baselineskip4\p@ \lineskiplimit\z@
    \hbox{.}\hbox{.}\hbox{.}}}
\makeatother

%% file: introduction.tex
\section{Introduction}
It is widely accepted that the existence of one-way functions (OWFs) is the minimal assumption 
in classical cryptography~\cite{FOCS:ImpLub89}.
On the other hand, in quantum cryptography, OWFs would not necessarily be the minimal assumption:
there could exist a new world, ``Microcrypt'', where quantum cryptography is possible without OWFs~\cite{Kre21,C:MorYam22,C:AnaQiaYue22}.
Many fundamental quantum cryptographic primitives and applications have been found in Microcrypt~\cite{C:JiLiuSon18,TCC:AGQY22,ITCS:BraCanQia23,TQC:MorYam24,STOC:KhuTom24,TCC:BGHMSV23,ITCS:AnaLinYue24,AC:MorYamYam24,TCC:BBSS23,cryptoeprint:2024/1578,STOC:MaHua25,C:ChuGolGra24,FOCS:BatJai24}
and separated from OWFs~\cite{Kre21,STOC:KQST23,STOC:LomMaWri24,STOC:KreQiaTal25}. 

Among these many fundamental primitives,
one-way puzzles (OWPuzzs) introduced by Khurana and Tomer~\cite{STOC:KhuTom24} are a natural quantum analogue of OWFs.
A OWPuzz is a pair $(\Samp,\Ver)$ of a quantum polynomial-time (QPT) sampling algorithm $\Samp$ and a (not-necessarily-efficient) verification algorithm $\Ver$.
$\Samp$ outputs two bit strings, $\ans$ and $\puzz$, that pass the verification $\Ver(\ans,\puzz)$ with high probability.
The security of OWPuzzs requires that no QPT adversary $\cA$ given $\puzz$ can output $\ans'$ that passes the verification $\Ver(\puzz,\ans')$ with high probability.
Almost all primitives in Microcrypt imply OWPuzzs and several important primitives and applications are implied by OWPuzzs including 
EFI pairs~\cite{ITCS:BraCanQia23,STOC:KhuTom24}, non-interactive commitments~\cite{AC:Yan22,C:MorYam22}, 
and multi-party computations~\cite{C:MorYam22,C:AnaQiaYue22,C:BCKM21b,EC:GLSV21}. 

OWPuzzs can be trivially constructed from (quantumly-secure) OWFs, but their construction based on other plausible assumptions, 
particularly those that do not imply OWFs, has not been extensively studied. (There are only three results. See \cref{sec:related}.)
One of the most significant goals in the field of quantum cryptography is to base OWPuzz on some plausible assumptions that will not imply OWFs.

\subsection{Our Results}
According to the laws of quantum physics, quantum states are generally collapsed by measurements.
However, we could imagine some magical {\it non-collapsing measurements} that can sample measurement results without collapsing quantum states.
The assumption that such highly-unphysical measurements are impossible in QPT seems extremely plausible.\footnote{If a single copy of unknown quantum state is given, non-collapsing measurements on the state
are statistically impossible (as long as we believe the standard quantum physics). On the other hand, as we will explain later, 
in our setup, a classical description of a quantum circuit that generates states 
is given, and therefore non-collapsing measurements are possible in unbounded time. Our assumptions are therefore computational ones.}
The contribution of this paper is to base quantum cryptography on such a ``physically reasonable'' assumption, directly motivated by a fundamental law of quantum physics.
Specifically, we show that OWPuzzs can be based on the computational hardness of simulating non-collapsing measurements.

\paragraph{Construction of OWPuzzs.}
To that end, we first introduce a new complexity class, $\mathbf{SampPDQP}$,
that characterizes a computational power of non-collapsing measurements.
We then show the following result.
\begin{theorem}
\label{thm:1}
If $\mathbf{SampPDQP}$ is hard on average,\footnote{
    We say that a complexity class $\mathbf{C}$ of sampling problems is hard on average   
    if there exist a sampling problem $\{\cD_x\}_{x}\in\mathbf{C}$, a polynomial $p$, and a QPT samplable distribution $\cE(1^\secp)\to x \in\bit^\secp$
    such that for any QPT algorithm $\cF$ and for all sufficiently large $\secp\in\N$,
    $\SD(\{x,\cF(x)\}_{x\gets\cE(1^\secp)},\{x,\cD_x\}_{x\gets\cE(1^\secp)}) > \frac{1}{p(\secp)}$.
    Here, $\SD$ is the statistical distance.
} then OWPuzzs exist.    
\end{theorem}

$\mathbf{SampPDQP}$
is the sampling version of
$\mathbf{PDQP}$.
Let us first explain $\mathbf{PDQP}$.
$\mathbf{PDQP}$ was introduced by Aaronson, Bouland, Fitzsimons, and Lee~\cite{ITCS:ABFL16}. 
$\mathbf{PDQP}$ is the class of decision problems that can be solved with a classical deterministic polynomial-time
algorithm that can make a single query to a magical oracle $\cQ$ that can perform non-collapsing measurements.
More precisely, $\cQ$ takes a classical description of a quantum circuit $(U_1,M_1,...,U_T,M_T)$ as input.
Here, for each $i\in[T]$, $U_i$ is an $\ell$-qubit unitary and 
$M_i\coloneqq\{P_j^i\}_{j}$ is a {\it collapsing} projection measurement on $m_i$ qubits, where $0\le m_i\le\ell$.
(When $m_i=0$, this means that $M_i$ does not do any measurement.)
The oracle $\cQ$ first generates $U_1|0^\ell\rangle$ and performs the {\it collapsing} projection measurement $M_1$ on the state.
Assume that the result $j_1$ is obtained. Then the post-measurement state is
$|\psi_1\rangle\coloneqq P_{j_1}^1U_1|0^\ell\rangle/\sqrt{\|P_{j_1}^1U_1|0^\ell\rangle\|^2}$.
Then the oracle $\cQ$ does a {\it non-collapsing} measurement on all qubits of $|\psi_1\rangle$ in the computational basis, and gets the measurement result $v_1\in\bit^\ell$.
Because this is a non-collapsing measurement,
this measurement does not collapse $|\psi_1\rangle$ to $|v_1\rangle$:
even after obtaining $v_1$, the state is still $|\psi_1\rangle$.
The oracle $\cQ$ then applies $U_2$ on $|\psi_1\rangle$, and performs the {\it collapsing} projection measurement
$M_2$ on $U_2|\psi_1\rangle$.
Assume that the result $j_2$ is obtained. Then the post-measurement state is
$|\psi_2\rangle\coloneqq P_{j_2}^2U_2|\psi_1\rangle/\sqrt{\|P_{j_2}^2U_2|\psi_1\rangle\|^2}$.
The oracle $\cQ$ again performs the non-collapsing computational-basis measurement on the all qubits of
$|\psi_2\rangle$ to get the result $v_2\in\bit^\ell$.
The oracle $\cQ$ then applies $U_3$ on $|\psi_2\rangle$, measures it with $M_3$, and 
performs the non-collapsing measurement on the post-measurement state of $M_3$ to get $v_3\in\bit^\ell$, and so on.
In this way, the oracle obtains $(v_1,...,v_T)$. The oracle finally outputs $(v_1,...,v_T)$.

$\mathbf{SampPDQP}$ is the sampling version of $\mathbf{PDQP}$.
In other words, $\mathbf{SampPDQP}$ is the class of sampling problems\footnote{A sampling problem is a family $\{\cD_x\}_x$ of distributions over bit strings. We say that a sampling problem $\{\cD_x\}_x$
is solved if all $\cD_x$ can be sampled.} that are solved by a classical deterministic polynomial-time algorithm
that can make a single query to the non-collapsing measurement oracle $\cQ$.
We can show the following lemma:
\begin{lemma}
If $\mathbf{PDQP}$ is hard on average,\footnote{
 We say that a complexity class $\mathbf{C}$ of decision problems is hard on average   
    if there exist a language $L\in\mathbf{C}$, a polynomial $p$, and a QPT samplable distribution $\cE(1^\secp)\to x \in\bit^\secp$
    such that for any QPT algorithm $\cF$ and for all sufficiently large $\secp\in\N$,
    $\Pr_{x\gets\cE(1^\secp)}[\cF(x)\neq L(x)]>1/p(\secp)$. 
} then $\mathbf{SampPDQP}$ is hard on average.    
\end{lemma}
We therefore obtain the following result as a corollary of \cref{thm:1}.
\begin{corollary}
\label{coro:1}
If $\mathbf{PDQP}$ is hard on average, then OWPuzzs exist.    
\end{corollary}

\paragraph{Construction of auxiliary-input OWPuzzs.}
We have seen that average-case hardness of $\mathbf{SampPDQP}$ implies OWPuzzs.
What happens for the worst-case hardness, $\mathbf{SampPDQP}\not\subseteq\mathbf{SampBQP}$?
We can actually show that such a worst-case hardness implies auxiliary-input OWPuzzs.\footnote{
    Roughly, an auxiliary-input OWPuzz is a pair $(\Samp,\Ver)$ of algorithms such that for any QPT algorithm $\cA$, there exists $x$ such that under the sampling $(\puzz,\ans)\gets\Samp(x)$, $\cA(x,\puzz)$ fails to find $\ans'$ that is accepted by $\Ver(x,\puzz,\ans')$.
} 
\begin{theorem}
\label{thm:2}
If $\mathbf{SampPDQP}\not\subseteq\mathbf{SampBQP}$, 
then auxiliary-input OWPuzzs exist.
\end{theorem}
We can consider a generalizaiton of
$\mathbf{SampPDQP}$, which we call $\mathbf{SampAdPDQP}$. 
We show that the worst-case hardness of $\mathbf{SampAdPDQP}$ is equivalent to that of $\mathbf{SampPDQP}$.
\begin{theorem}
    \label{thm:worst}
    $\mathbf{SampAdPDQP}\nsubseteq\mathbf{SampBQP}$ if and only if $\mathbf{SampPDQP}\nsubseteq\mathbf{SampBQP}$.
\end{theorem}
$\mathbf{SampAdPDQP}$ is a generalization of
$\mathbf{SampPDQP}$. In $\mathbf{SampPDQP}$, the base algorithm can query the non-collapsing measurement oracle $\cQ$ only once, but
in $\mathbf{SampAdPDQP}$, it can query many times adaptively.\footnote{The non-collapsing measurement oracle is stateless, which means that it does not
keep its internal quantum state between queries.}
By combining \cref{thm:2,,thm:worst}, we obtain the following corollary.
\begin{corollary}
\label{cor:4}
If $\mathbf{SampAdPDQP}\not\subseteq\mathbf{SampBQP}$, 
then auxiliary-input OWPuzzs exist.
\end{corollary}



\paragraph{Relations to OWFs.}
\mor{It seems that Eli-Saachi commitment (where sending $y$ is done at the beginning of the commit phase) imply dCRPuzz. 
Because such commitment exists even if BQP=QCMA, this means that dCRPuzzs and therefore
our assumption will not imply OWFs.}
Although there is no formal proof, our assumptions do not seem to imply OWFs, 
because the ability of inverting classical functions does not seem to be useful to
realize quantum non-collapsing measurements.
Moreover, the following argument also suggests that
the average-case hardness of $\mathbf{SampPDQP}$ will not imply auxiliary-input OWFs (and therefore OWFs as well):
\cite{STOC:KreQiaTal25} left an open problem to separate quantum-evaluation collision-resistant hashing (CRH)
from $\mathbf{P}=\mathbf{NP}$ 
and gave a concrete candidate construction for it relative to an oracle. 
If this open problem is resolved, average-case hardness of $\mathbf{SampPDQP}$ does not imply auxiliary-input OWFs,
because (as we will see later) quantum-evaluation CRH implies average-case hardness 
of $\mathbf{SampPDQP}$, and auxiliary-input OWFs imply $\mathbf{P}\neq\mathbf{NP}$.

\paragraph{Distributional collision-resistant puzzles.}
We also study upperbounds of the hardness of $\mathbf{SampPDQP}$.
We introduce a new primitive, {\it distributional collision-resistant puzzles} (dCRPuzzs).
They are a natural quantum analogue of 
distributional collision-resistant hashing (dCRH)~\cite{STOC:DubIsh06}.\footnote{
We could also explore a quantum analogue of CRH.
However, its definition is not so straightforward.
For example, we could define a ``collision-resistant puzzle'' 
$(\Setup,\Samp,\Ver)$ 
as follows:
$\Setup$ is a QPT algorithm that, on input the security parameter $\secp$, outputs a public parameter $\pp$.
$\Samp$ is a QPT algorithm that, on input $\pp$, outputs two classical bit strings, $\puzz$ and $\ans$.
$\Ver$ is a (not-necessarily-efficient) algorithm that, on input $\pp$, $\puzz$, and $\ans$, outputs $\top/\bot$.
The ``collision-resistance'' requires that
no QPT adversary $\cA$ that receives $\pp$ as input can output $(\puzz,\ans,\ans')$ such that
$\ans\neq\ans'$ and
both $(\puzz,\ans)$ and $(\puzz,\ans')$ are accepted by $\Ver$ with high probability.
However, even if we require that the length of $\puzz$ is shorter than that of $\ans$, there is a trivial statistically-secure construction:
$\Samp$ always outputs $\puzz=0$ and $\ans=00$.
$\Ver$ accepts only $(\puzz=0,\ans=00)$.
}

A dCRPuzz is a set $(\Setup,\Samp)$ of algorithms. 
$\Setup$ is a QPT algorithm that, on input the security parameter $\secp$, outputs a public parameter $\pp$.
$\Samp$ is a QPT algorithm that, on input $\pp$, outputs two classical bit strings, $\puzz$ and $\ans$.
The security requirement is that for any QPT adversary $\cA$, the statistical distance between two distributions,
$(\pp,\cA(\pp))_{\pp\gets\Setup(1^\secp)}$
and
$(\pp,\mathsf{Col}(\pp))_{\pp\gets\Setup(1^\secp)}$,
is large. Here, $\mathsf{Col}(\pp)\to(\puzz,\ans,\ans')$ is the following distribution:
It first samples $(\puzz,\ans)\gets\Samp(\pp)$, and then samples $\ans'$ with the conditional probability
$\Pr[\ans'|\puzz]=\Pr[(\ans',\puzz)\gets\Samp(\pp)]/\Pr[\puzz\gets\Samp(\pp)]$.

It is trivial that (quantumly-secure) dCRH (and therefore collision-resistant hashing (CRH) as well) imply dCRPuzzs.
Moreover, because an average-case hardness of $\mathbf{SZK}$ for QPT implies quantumly-secure dCRH~\cite{C:KomYog18}, 
the average-case hardness of $\mathbf{SZK}$ for QPT also implies dCRPuzzs.

We show that dCRPuzzs are an upperbound of the hardness of $\mathbf{SampPDQP}$:
\begin{theorem}
If dCRPuzzs exist, then $\mathbf{SampPDQP}$ is hard on average.    
\end{theorem}

\paragraph{Applications that imply dCRPuzzs.}
We show that several applications imply dCRPuzzs.
We first show that one-shot message authentication codes (MACs) 
imply dCRPuzzs.
\begin{theorem}
If one-shot MACs exist, then dCRPuzzs exist.  
\end{theorem}
A one-shot signature~\cite{STOC:AGKZ20} is a digital signature scheme with a quantum signing key.
Signing a message with the key can be done only once.
One-shot MACs~\cite{untele} are a privately-verifiable version of one-shot signatures.
One-shot MACs are also a relaxation of two-tier one-shot signatures~\cite{cryptoeprint:2023/1937},
where the verification is partially public.
Because two-tier one-shot signatures can be constructed from the LWE assumption~\cite{cryptoeprint:2023/1937},
one-shot MACs can also be constructed from the LWE assumption~\cite{untele}.

We also show that two-message honest-statistically-hiding commitments with classical communication
imply dCRPuzzs.\footnote{Here, honest statistical-hiding means that the adversary behaves honestly in the commit phase.}
\begin{theorem}
If two-message honest-statistically-hiding commitments with classical communication exist, then dCRPuzzs exist.
\end{theorem}

\paragraph{Summary of our results.}
Finally, all our results obtained in this paper and known results are summarized in \cref{fig:graph}.
\input{figure}

\subsection{Related Works}
\label{sec:related}


\paragraph{Basing OWPuzzs on plausible assumptions.}
Recently, several results have been obtained that base OWPuzzs on some plausible assumptions that do not seem to imply OWFs.
Khurana and Tomer~\cite{STOC:KhuTom25} showed that OWPuzzs can be constructed
from the assumption of $\mathbf{P}^{\mathbf{PP}}\not\subseteq\mathbf{BQP}$ plus
certain hardness assumptions that imply sampling-based quantum advantage. 
Hiroka and Morimae~\cite{C:HirMor25}, and 
Cavalar, Goldin, Gray and Hall~\cite{EC:CGGH25}
independently showed that the existence of OWPuzzs is equivalent to certain average-case hardness of estimating Kolmogorov complexity.
We do not know if their assumptions are related to our generic assumption of average-case hardness of $\mathbf{(Samp)PDQP}$. 
It would be interesting if there are some relations.

\paragraph{Relations to $\mathbf{SZK}$ and $\mathbf{CZK}$.}
It is well known in classical cryptography
that (classically-secure) OWFs exist if $\mathbf{SZK}$ is hard on average for probabilistic polynomial-time (PPT),
and that (classically-secure) auxiliary-input OWFs
 exist if $\mathbf{SZK}\not\subseteq\mathbf{BPP}$~\cite{CCC:Ost91}.
Their proofs can be easily extended to show that
quantumly-secure OWFs exist if $\mathbf{SZK}$ is hard on average for QPT
and that
quantumly-secure auxiliary-input OWFs exist if $\mathbf{SZK}\not\subseteq\mathbf{BQP}$.
Because quantumly-secure (auxiliary-input) OWFs imply (auxiliary-input) OWPuzzs, this means that
\begin{corollary}[\cite{CCC:Ost91}]
\label{coro:trivial1}
If $\mathbf{SZK}$ is hard on average for QPT, then OWPuzzs exist. 
\end{corollary}
\begin{corollary}[\cite{CCC:Ost91}]
\label{coro:trivial2}
If $\mathbf{SZK}\not\subseteq\mathbf{BQP}$,
then auxiliary-input OWPuzzs exist.
\end{corollary}
Because $\mathbf{SZK}\subseteq\mathbf{PDQP}$~\cite{ITCS:ABFL16},
our results, \cref{thm:1,cor:4}, are improvements of \Cref{coro:trivial1,coro:trivial2}.

\cite{ISTCS:OstWig93} improved the results of \cite{CCC:Ost91} to $\mathbf{CZK}$, and their proofs can be easily extended to
the quantum case as well. We therefore obtain
\begin{corollary}[\cite{ISTCS:OstWig93}]
\label{coro:trivial3}
If $\mathbf{CZK}$ is hard on average for QPT, then OWPuzzs exist. 
\end{corollary}
\begin{corollary}[\cite{ISTCS:OstWig93}]
\label{coro:trivial4}
If $\mathbf{CZK}\not\subseteq\mathbf{BQP}$,
then auxiliary-input OWPuzzs exist.
\end{corollary}
To our knowledge, there is no known relation between $\mathbf{PDQP}$ and $\mathbf{CZK}$, and therefore our 
results, \cref{thm:1,cor:4}, are incomparable to \cref{coro:trivial3,coro:trivial4}.

\paragraph{Relations to $\mathbf{PP}$.}
Khurana and Tomer~\cite{STOC:KhuTom25} 
constructed OWPuzzs from
the assumption of $\mathbf{P}^{\mathbf{PP}}\not\subseteq\mathbf{BQP}$
plus some assumptions on which sampling-based quantum advantage are based.
They left the
question of
whether OWPuzzs can be based solely on 
$\mathbf{P}^{\mathbf{PP}}\not\subseteq\mathbf{BQP}$ or its average-case version.
Because $\mathbf{PDQP}\subseteq \mathbf{P}^{\mathbf{PP}}$~\cite{ITCS:ABFL16}, and $\mathbf{SampAdPDQP}\nsubseteq\mathbf{SampBQP}$ implies $\mathbf{P}^{\mathbf{PP}} \nsubseteq \mathbf{BQP}$, 
our results, \cref{thm:1,cor:4}, 
solve weaker versions of their open problem.

\subsection{Technical Overview}
In this subsection, we provide high-level overviews of our proofs.

\paragraph{OWPuzzs from the average-case hardness of $\mathbf{SampPDQP}$.}
We first explain our construction of OWPuzzs from average-case hardness of $\mathbf{SampPDQP}$.
Our proof technique is inspired by the notion of universal extrapolation \cite{FOCS:ImpLev90,CCC:Ost91} and its application in the quantum setting \cite{STOC:KhuTom25,C:HirMor25,EC:CGGH25}.
Because of the equivalence between OWPuzzs and distributional OWPuzzs (DistOWPuzzs) \cite{C:ChuGolGra24}, it suffices to construct DistOWPuzzs.
Here a DistOWPuzz is a QPT sampling algorithm $\Samp(1^\secp)\to(\puzz,\ans)$ that satisfies the following:
There exists a polynomial $q$ such that for any QPT algorithm $\cA$,
\begin{align}
\SD(\{\puzz,\ans\}_{(\puzz,\ans)\gets\Samp(1^\secp)}
,\{\puzz,\cA(\ans)\}_{(\puzz,\ans)\gets\Samp(1^\secp)}
)>\frac{1}{q(\secp)}
\end{align}
for all sufficiently large $\secp\in\mathbb{N}$. 

Assume that $\mathbf{SampPDQP}$ is hard on average.
This means that there exist a sampling problem $\{\cD_x\}_x\in\mathbf{SampPDQP}$, a polynomial $p$ and a QPT algorithm $\cE(1^\secp)\to x\in\bit^\secp$ 
such that for any QPT algorithm $\cF$,
\begin{align}\label{eq:overview_1}
\SD(\{x,\cF(x)\}_{x\gets\cE(1^\secp)},\{x,\cD_x\}_{x\gets\cE(1^\secp)})>\frac{1}{p(\secp)}    
\end{align}
for all sufficiently large $\secp\in\mathbb{N}$.
From this $\cE$, we construct a DistOWPuzzs $\Samp$ as follows.
\begin{enumerate}
\item
Sample $x\gets\cE(1^\secp)$.
\item 
Let 
$C_x:=(U_1,M_1,...,U_T,M_T)$ 
be (a classical description of) a quantum circuit 
that is queried to $\cQ$
corresponding to the instance $x$.
    \item 
    Choose $i\gets[T]$.
    \item 
    Run $(U_1,M_1,...,U_i,M_i)$ to get the measurement results $(u_1,...,u_i)$,
    where $u_i$ is the measurement result of the measurement $M_i$.
    \item 
    Measure all qubits of the resulting state to get the result $v_i=u_i\|w_i$.
   \item 
   Output $\puzz\coloneqq(x,i,u_1,...,u_i)$ and $\ans\coloneqq w_i$.
\end{enumerate}
For the sake of contradiction, assume that $\Samp$ is not a DistOWPuzz.
Then, for any polynomial $q$, there exists a QPT $\cA$ such that 
\begin{align}\label{eq:overview_2}
    \SD (\{x,i,u_1,...,u_i,w_i\} , \{x,i,u_1,...,u_i,\cA(x,i,u_1,...,u_i)\}) \le \frac{1}{q(\secp)}
\end{align}
for infinitely many $\secp\in\mathbb{N}$, where $(x,i,u_1,...,u_i,w_i)\gets\Samp(1^\secp)$.
Our goal is to construct a QPT algorithm that contradicts \cref{eq:overview_1}.
Define the QPT algorithm $\cB$ that simulates the output distribution of $\cQ$ as follows:
\begin{enumerate}
    \item Take $x$ and $C_x=(U_1,M_1,...,U_T,M_T)$ as input.
    \item Run $C_x$ and obtain $(u_1,...,u_T)$, where $u_i$ is the outcome of the measurement $M_i$.
    \item For all $i\in[T]$, run $w_i\gets\cA(x,i,u_1,...,u_i)$.
    \item Output $(u_1\|w_1,...,u_T\|w_T)$.
\end{enumerate}
Roughly speaking, because of \cref{eq:overview_2}, the distribution $w_i\gets\cA(x,i,u_1,...,u_i)$
in the third step of $\cB$ is close to the distribution $\Pr[w_i\gets\cQ(C_x)|x\gets\cE(1^\secp),i\gets[T],(u_1,...,u_i)\gets\cQ(C_x)]$.
Therefore, the output distribution of $\cB$ is close to that of $\cQ$.
Hence a QPT algorithm that runs the base machine of $\mathbf{SampPDQP}$, which is a polynomial-time deterministic machine, 
and runs the QPT algorithm $\cB$ instead of the query to $\cQ$ breaks \cref{eq:overview_1}.

\paragraph{Auxiliary-input OWPuzzs from the worst-case hardness of $\mathbf{SampAdPDQP}$.}
Our second result is the construction of auxiliary-input OWPuzzs from the worst-case assumption.\footnote{Here,  we introduce an idea that directly proves \cref{cor:4}. The proof of \cref{thm:2} follows as a special case of the same technique.}
The basic idea of the proof is similar to that of the first result,
but there are two crucial issues, and we need more careful investigations.
The first issue is that the assumption is now the worst-case hardness.
The second issue is that now adaptive queries are allowed.
Primitives constructed from worst-case assumptions often have to be auxiliary-input ones, and in fact, this is also the case here:
what we construct is an auxiliary-input OWPuzzs.
A slightly non-trivial and an interesting point is that the second issue is also resolved
by considering only the auxiliary-input situation!

The first issue is easily resolved by giving the instance $x$ to the OWPuzz as input. 
Let us explain more details about the second issue.
When adaptive queries are allowed,
the second query to $\cQ$ can be a bit string that is not necessarily QPT generatable,
because the second query can depend on the non-collapsing measurement results done in
the first query to $\cQ$. 
We resolve the issue by providing the answers of previous queries as auxiliary-input.
Our construction of auxiliary-input DistOWPuzzs $\Samp$ is as follows:
\begin{enumerate}
    \item Take $x\in\bit^*$, $k\in\N$, and a collection $(V_1,...,V_k)$ of $k$ bit strings as input.
    \item Let $C_{x,k+1}:=(U_1,M_1,...,U_T,M_T)$ be (a classical description of) a quantum circuit 
    that is the $(k+1)$-th query to $\cQ$ corresponding to the instance $x$. 
    This is generated in polynomial-time by running the base machine of $\mathbf{SampAdPDQP}$ and 
    using $(V_1,...,V_k)$ as answers of the previous $k$ queries.
    \item Choose $i\gets[T]$.
    \item Run $(U_1,M_1,...,U_i,M_i)$ to get the measurement results $(u_1,...,u_i)$,
    where $u_i$ is the measurement result of the measurement $M_i$.
    \item Measure all qubits of the resulting state to get the result $v_i=u_i\|w_i$.
   \item Output $\puzz\coloneqq(i,u_1,...,u_i)$ and $\ans\coloneqq w_i$.
\end{enumerate}
Then, we can show the result in a similar way as the average case.

\paragraph{dCRPuzzs imply average-case hardness of $\mathbf{SampPDQP}$.}
Remember that a dCRPuzz is a pair $(\Setup,\Samp)$ of QPT algorithms such that
for any QPT adversary $\cA$, the statistical distance between two distributions,
$(\pp,\cA(\pp))_{\pp\gets\Setup(1^\secp)}$
and
$(\pp,\mathsf{Col}(\pp))_{\pp\gets\Setup(1^\secp)}$,
is large. Here, $\mathsf{Col}(\pp)\to(\puzz,\ans,\ans')$ is the following distribution:
It first samples $(\puzz,\ans)\gets\Samp(\pp)$, and then samples $\ans'$ with the conditional probability
$\Pr[\ans'|\puzz]=\Pr[(\ans',\puzz)\gets\Samp(\pp)]/\Pr[\puzz\gets\Samp(\pp)]$.
Without loss of generality, we can assume that $\Samp$ runs as follows:
\begin{enumerate}
    \item 
    Apply a unitary $V_\pp$ on $|0...0\rangle$ to generate $V_\pp|0...0\rangle=\sum_{\puzz,\ans} c_{\puzz,\ans} |\puzz\rangle|\ans\rangle$.
    \item 
    Measure the first register to get $\puzz$.
    \item 
    Measure the second register to get $\ans$.
    \item 
    Output $(\puzz,\ans)$.
\end{enumerate}

It is easy to see that the dCRPuzz is broken by querying the following $C=(U_1,M_1,U_2,M_2)$ to $\cQ$:
\begin{enumerate}
\item 
$U_1=V_\pp$ 
    \item $M_1$ is the measurement on the first register.
    \item $U_2=I$.
    \item $M_2$ does not do any measurement.
\end{enumerate}

\paragraph{One-shot MACs imply dCRPuzzs.}
Let $m_0$ and $m_1$ be any two different classical messages.
Let $\vk$ be a verification key and $|\sigk\rangle$ be a quantum signing key.
If we apply the signing algorithm $\Sign$ coherently
on
$(|m_0\rangle+|m_1\rangle)|\sigk\rangle|0...0\rangle$
and measure the last register, we get
$(m_0,\sigma_0)$ or $(m_1,\sigma_1)$,
where, for each $b\in\bit$, $\sigma_b$ is a valid signature for $m_b$.
If we consider $\vk$ as $\puzz$ and $(m_b,\sigma_b)$ as $\ans$, non-existence of dCRPuzzs means that
we can sample $(m_b,\sigma_b)$ twice with the same $\vk$.
Then, with probability at least 1/2, we get
both
$(m_0,\sigma_0)$ and
$(m_1,\sigma_1)$.
Hence the one-shot MAC is broken.

\paragraph{Commitments imply dCRPuzzs.}
We can also show that
two-message honest-statistically-hiding bit commitments with classical communication imply dCRPuzzs.
In the two-message commitment, the sender receives a message $r$
from the receiver, and then returns $\com$, which is the commitment.
Assume that sender commits both 0 and 1 in superposition.
This means that the sender applies the commitment algorithm coherently on $(|0\rangle+|1\rangle)|0...0\rangle$
and measures the second register to get the commitment $\com$.
Because of the statistical hiding, the state after the measurement is close to
$|0\rangle|\mathsf{decom}_0\rangle|junk_0\rangle+|1\rangle|\mathsf{decom}_1\rangle|junk_1\rangle$,
where $\mathsf{decom}_b$ is the decommitment for bit $b$.
Consider $\com$ as $\puzz$ and $(b,\mathsf{decom}_b)$ as $\ans$.
Then if dCRPuzzs do not exist, we can sample both $(0,\mathsf{decom}_0)$
and $(1,\mathsf{decom}_1)$ with the same $\com$, which breaks the binding.

%% file: figure.tex
\usetikzlibrary{positioning} 
\usetikzlibrary{calc} 
\usetikzlibrary {quotes}
\tikzset{>=latex} 

\tikzstyle{mysmallarrow}=[->,black,line width=1.6]
\tikzstyle{mydbarrow}=[<->,black,line width=1.6]
\tikzstyle{newarrow}=[->,red,line width=1.6]
\tikzstyle{newsinglearrow}=[->,red,line width=1.6]
\tikzstyle{carrow}=[->,red,line width=1.6]
\begin{figure}
\begin{center}
    \begin{tikzpicture}[scale=0.9,every edge quotes/.style = {font=\footnotesize,fill=white}]
      \def\h{-2.0} 
      \def\w{2.6} 

        \node[] (PP) at (3.3*\w,5.8*\h) {$\mathbf{P}^\mathbf{PP}\not\subseteq \mathbf{BQP}$};
       \node[] (SZK) at (5*\w,0.1*\h) {AH of $\mathbf{SZK}$};
       \node[] (SampPDQP) at (5*\w,2.8*\h) {AH of $\mathbf{SampPDQP}$};
       \node[] (worstPDQP) at (3.3*\w,2.6*\h) {$\mathbf{PDQP}\not\subseteq\mathbf{BQP}$};
       \node[] (worstSampAdPDQP) at (3.3*\w,3.5*\h) {$\mathbf{Samp(Ad)PDQP}\not\subseteq\mathbf{SampBQP}$};
       \node[] (aiOWPuzzs) at (3.3*\w,4.8*\h) {aiOWPuzzs};
       \node[] (AHPDQP) at (5*\w,1.3*\h) {AH of $\mathbf{PDQP}$};
       \node[] (OWFs) at (6.2*\w,2.8*\h) {OWFs};
       \node[] (OWPuzzs) at (5*\w,4.2*\h) {OWPuzzs};
        \node[] (CZK) at (6.2*\w,1.6*\h) {AH of $\mathbf{CZK}$};
       \node[] (CRH) at (3.8*\w,0.5*\h) {CRH};
       \node[] (dCRH) at (3.8*\w,1.2*\h) {dCRH};
       \node[] (dCRPs) at (3.8*\w,1.8*\h) {dCRPuzzs};
       \node[] (one-shot signatures) at (2.9*\w,0.4*\h) {one-shot signatures};
       \node[] (one-shot MACs) at (2.9*\w,1*\h) {one-shot MACs};
       \node[] (commitments) at (2.2*\w,1.4*\h) {2 HSH Comm CC};


        \draw[newarrow] (commitments) edge[] (dCRPs);
        \draw[mysmallarrow=black] (AHPDQP) edge[] (worstPDQP);
        \draw[mysmallarrow=black] (worstPDQP) edge[] (worstSampAdPDQP);
        \draw[newarrow] (worstSampAdPDQP) edge[] (aiOWPuzzs);
        \draw[mysmallarrow=black] (SampPDQP) edge[] (worstSampAdPDQP);
        \draw[mysmallarrow=black] (SZK) edge["\cite{ITCS:ABFL16}"] (AHPDQP);
        \draw[mysmallarrow=black] (AHPDQP) edge[] (SampPDQP);
        \draw[mysmallarrow=black] (one-shot signatures) edge[] (one-shot MACs);
        \draw[newarrow] (one-shot MACs) edge[] (dCRPs);
        \draw[mysmallarrow=black] (dCRH) edge[] (OWFs);
        \draw[mysmallarrow=black] (CRH) edge[] (dCRH);
        \draw[mysmallarrow=black] (dCRH) edge[] (dCRPs);
        \draw[newarrow=red] (dCRPs) edge[] (SampPDQP);
        \draw[mysmallarrow=black] (SZK) edge["\cite{C:KomYog18}"] (dCRH);
        \draw[newarrow] (SampPDQP) edge[] (OWPuzzs);
        \draw[mysmallarrow=black] (OWFs) edge[] (OWPuzzs);
        \draw[mysmallarrow=black] (SZK) edge[] (CZK);
        \draw[mysmallarrow=black] (CZK) edge["\cite{ISTCS:OstWig93}"] (OWFs);
        \draw[mysmallarrow=black] (OWPuzzs) edge[] (aiOWPuzzs);
        \draw[mysmallarrow=black] (aiOWPuzzs) edge["\cite{cavalar2023}"] (PP);
    \end{tikzpicture}
\end{center}
\caption{A summary of results. 
Black lines are known results or trivial implications.
Red lines are new in our work. 
``AH'' stands for the average-case hardness.
``ai'' stands for auxiliary-input.
``2 HSH Comm CC'' is two-message honest-statistically-hiding commitments with classical communication.
}
\label{fig:graph}
\end{figure}
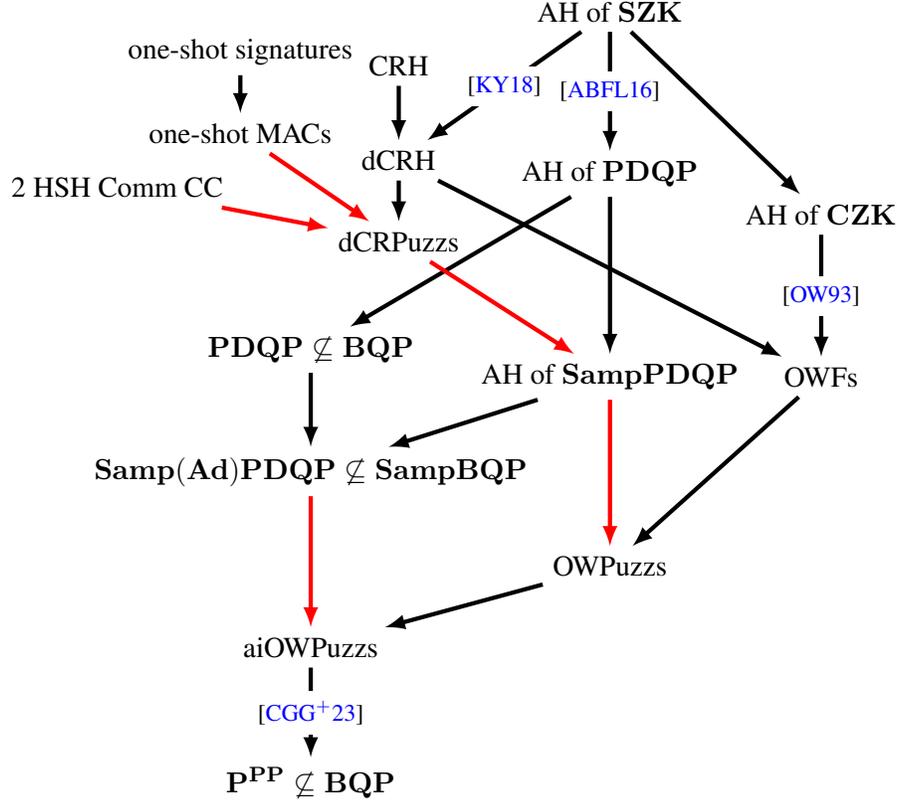

%% file: preliminaries.tex
\section{Preliminaries}\label{sec:preliminaries}
\subsection{Basic Notations} 
We use standard notations of quantum computing and cryptography.
For a bit string $x$, $|x|$ is its length.
For bit strings $x$ and $y$, $x\|y$ is their concatenation.
$\mathbb{N}$ is the set of natural numbers.
We use $\secp$ as the security parameter.
$[n]$ means the set $\{1,2,...,n\}$.
For a finite set $S$, $x\gets S$ means that an element $x$ is sampled uniformly at random from the set $S$.
$\negl$ is a negligible function, and $\poly$ is a polynomial.
All polynomials appear in this paper are positive, but for simplicity we do not explicitly mention it.
PPT stands for (classical) probabilistic polynomial-time and QPT stands for quantum polynomial-time. 
For an algorithm $\cA$, $y\gets \cA(x)$ means that the algorithm $\cA$ outputs $y$ on input $x$.
If $\cA$ is a classical probabilistic or quantum algorithm that takes $x$ as input and outputs bit strings,
we often mean $\cA(x)$ by the output probability distribution of $\cA$ on input $x$.
For two probability distributions $P\coloneqq\{p_i\}_i$ and $Q\coloneqq\{q_i\}_i$, 
$\SD(Q,P)\coloneqq\frac{1}{2}\sum_i|p_i-q_i|$ is their statistical distance.

\subsection{One-Way Puzzles and Distributional One-Way Puzzles}

We first review the definition of one-way puzzles (OWPuzzs).

\begin{definition}[OWPuzzs \cite{STOC:KhuTom24}]
    A one-way puzzle (OWPuzz) is a pair $(\Samp, \Ver)$ of algorithms such that
    \begin{itemize}
    \item 
    $\Samp(1^\secp)\to (\puzz,\ans):$
    It is a QPT algorithm that, on input the security parameter $\secp$, 
    outputs a pair 
    $(\puzz,\ans)$
    of classical strings.
    \item
    $\Ver(\puzz,\ans')\to\top/\bot:$
    It is an unbounded algorithm that, on input $(\puzz,\ans')$, outputs either $\top/\bot$.
    \end{itemize}
    They satisfy the following properties.
    \begin{itemize}
    \item {\bf Correctness:} 
    \begin{align}
        \Pr[\top\gets\Ver(\puzz,\ans):
        (\puzz,\ans)\gets\Samp(1^\secp)]
        \ge 1-\negl(\secp).
    \end{align}
    \item {\bf Security:} For any QPT adversary $\cA$,
    \begin{align}
        \Pr[\top\gets\Ver(\puzz,\cA(1^\secp,\puzz)):
        (\puzz,\ans)\gets \Samp(1^\secp)] \le \negl(\secp).
    \end{align}
    \end{itemize}
\end{definition}

We also review the definition of distributional one-way puzzles (DistOWPuzzs).

\begin{definition}[DistOWPuzzs \cite{C:ChuGolGra24}]
    A uniform QPT algorithm $\Samp$ that takes the security parameter $1^\secp$ as input and outputs a pair $(\puzz,\ans)$ of bit strings
    is called an $\alpha$-distributional one-way puzzle ($\alpha$-DistOWPuzz) 
    if there exists a function $\alpha:\N\to[0,1]$ such that for any QPT adversary $\cA$, 
    and for all sufficiently large $\secp\in\N$,
    \begin{align}
        \SD\left( \{\puzz,\ans\}_{(\puzz,\ans)\gets\Samp(1^\secp)} , \{\puzz,\cA(1^\secp,\puzz)\}_{(\puzz,\ans)\gets\Samp(1^\secp)} \right) \ge \alpha(\secp).
    \end{align}
    If $\Samp$ is a $\secp^{-c}$-DistOWPuzz for some constant $c>0$, we simply say $\Samp$ is a DistOWPuzz.
\end{definition}

Clearly, if $(\Samp,\Ver)$ is a OWPuzz, then $\Samp$ is a DistOWPuzz.
Chung, Goldin, and Gray \cite{C:ChuGolGra24} showed that DistOWPuzzs imply OWPuzzs.
Combining them, the following equivalence is known.

\begin{lemma}[\cite{C:ChuGolGra24}]
    \label{lem:DistOWPuzz_OWPuzz}
    OWPuzzs exist if and only if DistOWPuzzs exist.
\end{lemma}

We define auxiliary-input variants of OWPuzzs and DistOWPuzzs.
\begin{definition}[Auxiliary-Input OWPuzzs]
    An auxiliary-input one-way puzzle (auxiliary-input OWPuzz) is a pair $(\Samp, \Ver)$ of algorithms such that
    \begin{itemize}
    \item 
    $\Samp(x)\to (\puzz,\ans):$
    It is a QPT algorithm that, on input a bit string $x$, 
    outputs a pair 
    $(\puzz,\ans)$
    of classical bit strings.
    \item
    $\Ver(x,\puzz,\ans')\to\top/\bot:$
    It is an unbounded algorithm that, on input $(x,\puzz,\ans')$, outputs either $\top/\bot$.
    \end{itemize}
    They satisfy the following properties.
    \begin{itemize}
    \item {\bf Correctness:} 
    \begin{align}
        \Pr[\top\gets\Ver(x,\puzz,\ans):
        (\puzz,\ans)\gets\Samp(x)]
        \ge 1-\negl(|x|).
    \end{align}
    \item {\bf Security:} For any QPT adversary $\cA$, there exists an infinite subset $I\subseteq\bit^*$ such that for all $x\in I$,
    \begin{align}
        \Pr[\top\gets\Ver(x,\puzz,\cA(x,\puzz)):
        (\puzz,\ans)\gets \Samp(x)] \le \negl(|x|).
    \end{align} 
    \end{itemize}
\end{definition}
\begin{definition}[Auxiliary-Input DistOWPuzzs]
    A uniform QPT algorithm $\Samp$ that takes an advice bit string $x\in\bit^*$ as input and outputs a pair $(\puzz,\ans)$ of bit strings
    is called an auxiliary-input $\alpha$-distributional one-way puzzle (auxiliary-input $\alpha$-DistOWPuzz) 
    if there exists a function $\alpha:\N\to[0,1]$ such that for any QPT adversary $\cA$,
    there exists an infinite subset $I\subseteq\bit^*$ such that for all $x\in I$,
    \begin{align}
        \SD\left( \{\puzz,\ans\}_{(\puzz,\ans)\gets\Samp(x)} , \{\puzz,\cA(x,\puzz)\}_{(\puzz,\ans)\gets\Samp(x)} \right) \ge \alpha(|x|).
    \end{align}
    If $\Samp$ is an auxiliary-input $\secp^{-c}$-DistOWPuzzs for some constant $c>0$, we simply say that $\Samp$ is an auxiliary-input DistOWPuzzs.
\end{definition}

The auxiliary-input version of \cref{lem:DistOWPuzz_OWPuzz} can be obtained by slightly modifying the proof of \cref{lem:DistOWPuzz_OWPuzz}.
\begin{lemma}[\cite{C:ChuGolGra24}]
    \label{lem:AI_DistOWPuzz_OWPuzz}
    Auxiliary-input OWPuzzs exist if and only if auxiliary-input DistOWPuzzs exist.
\end{lemma}

%% file: sampling.tex
\section{$\mathbf{SampPDQP}$}
\label{sec:sampling}
In this section, we introduce a sampling version of $\mathbf{PDQP}$, which we call $\mathbf{SampPDQP}$. 

\subsection{PDQP}
Before introducing $\mathbf{SampPDQP}$,
we review the definition of the decision class $\mathbf{PDQP}$ introduced in \cite{ITCS:ABFL16}. 
The complexity class $\mathbf{PDQP}$ is a class of decision problems
that can be solved with a polynomial-time classical deterministic algorithm that has a single query access 
to the non-collapsing measurement oracle. We first define
the non-collapsing measurement oracle.

\begin{definition}[Non-Collapsing Measurement Oracle \cite{ITCS:ABFL16}]
    A non-collapsing measurement oracle $\cQ$ is an oracle that behaves as follows:
    \begin{enumerate}
        \item Take (a classical description of) a quantum circuit $C=(U_1,M_1,...,U_T,M_T)$ and an integer $\ell>0$ as input. 
        Here each $U_i$ is a unitary operator on $\ell$ qubits and each $M_i$ is a computational-basis projective measurement on $m_i$ qubits such that $0\le m_i\le\ell$. (When $m_i=0$, this means that no measurement is done.)
        \item Let $|\psi_0\rangle:=|0^\ell\rangle$. 
        Run $C$ on input $|\psi_0\rangle$,
        and obtain $(u_1,...,u_T)$, where
        $u_t\in\bit^{m_t}$ is the outcome of the measurement $M_t$ for each $t\in[T]$.
        Let $\tau_t:=(u_1,...,u_t)$.
        For each $t\in [T]$, 
        let $|\psi_t^{\tau_t}\rangle$ be the (normalized) post-measurement state immediately after the measurement $M_t$, i.e., 
        \begin{align}
            |\psi_t^{\tau_t}\rangle := \frac{ (|u_t\rangle\langle u_t|\otimes I)U_t|\psi_{t-1}^{\tau_{t-1}}\rangle }{ \sqrt{ \langle \psi_{t-1}^{\tau_{t-1}}|U_t^\dagger (|u_t\rangle\langle u_t|\otimes I)U_t|\psi_{t-1}^{\tau_{t-1}}\rangle } }.
        \end{align}
        \item 
        For each $t\in[T]$, 
        sample 
        $v_t\in\bit^\ell$ with probability
        $|\langle v_t|\psi_t^{\tau_t}\rangle|^2$.
        \item Output $(v_1,...,v_T)$.
    \end{enumerate}
\end{definition}
\begin{remark}
Note that each non-collapsing measurement is done on {\it all} qubits including those that have been measured by the previous collapsing measurement. Therefore, 
the measurement result $v_i$ of the $i$th non-collapsing measurement is written as $v_i=u_i\|w_i$ with a bit string $w_i\in\bit^{\ell-m_i}$,
where $u_i$ is the measurement result of the collapsing measurement $M_i$.
For example, after the measurement of $M_i$, the entire state becomes $|u_i\rangle\otimes|\phi_i\rangle$ with a certain $(\ell-m_i)$-qubit state $|\phi_i\rangle$, where $u_i$ is the measurement result of $M_i$.
Then the non-collapsing measurement measures all qubits. The measurement result on the first register is always $u_i$, and
therefore the measurement result $v_i$ of the non-collapsing measurement is always written as $v_i=u_i\|w_i$, where $w_i\in\bit^{\ell-m_i}$ 
is the measurement result of the non-collapsing measurement on $|\phi_i\rangle$.
\end{remark}

With the non-collapsing measurement oracle, the class $\mathbf{PDQP}$ is defined as follows.
\begin{definition}[PDQP \cite{ITCS:ABFL16}]
    A language $L$ is in $\mathbf{PDQP}$ if there exists a polynomial-time classical deterministic Turing machine $R$ with a single query to a non-collapsing measurement oracle $\cQ$ such that
    \begin{itemize}
        \item For all $x\in L$, $\Pr[1\gets R^\cQ(x)]\ge \alpha(|x|)$,
        \item For all $x\notin L$, $\Pr[1\gets R^\cQ(x)]\le \beta(|x|)$,
    \end{itemize}
    where $\alpha,\beta$ are functions such that $\alpha(|x|)-\beta(|x|)\ge1/\poly(|x|)$.
\end{definition}

\begin{remark}
    Note that the error bound $(\alpha,\beta)$ can be amplified to $(1-\negl,\negl)$ by the repetition \cite{ITCS:ABFL16}.
\end{remark}

We also use the notion of average-case hardness.
\begin{definition}[Average-Case Hardness of $\mathbf{PDQP}$]
We say that $\mathbf{PDQP}$ is hard on average if the following is satisfied:
there exist a language $L\in\mathbf{PDQP}$, a polynomial $p$, and a QPT algorithm $\cE(1^\secp)\to\bit^\secp$
such that for any QPT algorithm $\cF$ and for all sufficiently large $\secp\in\N$,
\begin{align}
    \Pr_{x\gets\cE(1^\secp)} [\cF(x)\neq L(x)] \ge \frac{1}{p(\secp)},
\end{align} 
where
\begin{align}
    L(x) := 
    \begin{cases}
        1 & x\in L \\ 
        0 & x\notin L.
    \end{cases}
\end{align}
\end{definition}


\subsection{$\mathbf{SampPDQP}$}
Next we define $\mathbf{SampPDQP}$. $\mathbf{SampPDQP}$ is the class of sampling problems that are solved with a polynomial-time classical deterministic algorithm
that can make a single query to the non-collapsing measurement oracle.
Sampling problems are defined as follows.
\begin{definition}[Sampling Problems~\cite{Aar14,ITCS:AarBuhKre24}]
\label{def:Samplingproblems}
A (polynomially-bounded) sampling problem is a collection $\{D_x\}_{x\in\bit^*}$ 
of probability distributions, 
where $D_x$ is a distribution over $\bit^{p(|x|)}$, for some fixed polynomial $p$.
\end{definition}

The sampling complexity class, $\mathbf{SampBQP}$, is defined as follows.
\begin{definition}[{\bf SampBQP}~\cite{Aar14,ITCS:AarBuhKre24}] 
\label{def:SampBQP}
{\bf SampBQP} is the class of (polynomially-bounded) sampling problems 
$\{D_x\}_{x\in\bit^*}$ for which there exists a QPT algorithm $\cB$ such that for all $x$ and all $\epsilon>0$, 
$\SD(\cB(x,1^{\lfloor 1/\epsilon \rfloor}),D_x) \le\epsilon$, 
where $\cB(x,1^{\lfloor 1/\epsilon\rfloor})$ is the output probability distribution of 
$\cB$ on input $(x, 1^{\lfloor 1/\epsilon\rfloor})$. 
\end{definition}

We define $\mathbf{SampPDQP}$ as follows.
\begin{definition}[{\bf SampPDQP}] 
\label{def:SampBQP}
{\bf SampPDQP} is the class of (polynomially-bounded) sampling problems 
$\{D_x\}_{x\in\bit^*}$ for which there exists a classical deterministic polynomial-time 
algorithm $\cB$ that makes a single query to the non-collapsing measurement oracle $\cQ$ such that for all $x$ and all $\epsilon>0$, 
$\SD(\cB(x,1^{\lfloor 1/\epsilon \rfloor}),D_x) \le\epsilon$, 
where $\cB(x,1^{\lfloor 1/\epsilon\rfloor})$ is the output probability distribution of 
$\cB$ on input $(x, 1^{\lfloor 1/\epsilon\rfloor})$. 
\end{definition}

We also use the notion of average-case hardness.
\begin{definition}[Average-case Hardness of $\mathbf{SampPDQP}$]
\label{def:AHSampPDQP}
    We say that $\mathbf{SampPDQP}$ is hard on average if the following is satisfied:
    there exist a sampling problem $\{\cD_x\}_x\in\mathbf{SampPDQP}$, a polynomial $p$, and a QPT algorithm $\cE(1^\secp)\to\bit^\secp$ 
    such that for any QPT algorithm $\cF$ and for all sufficiently large $\secp$,
    \begin{align}
        \SD ( \{x,\cF(x)\}_{x\gets\cE(1^\secp)}, \{x,\cD_x\}_{x\gets\cE(1^\secp)} ) > \frac{1}{p(\secp)}.
    \end{align}
\end{definition}

We show the following lemma.
\begin{lemma}\label{lem:HoA_decision_samp}
    If $\mathbf{PDQP}$ is hard on average, then $\mathbf{SampPDQP}$ is hard on average.
\end{lemma}
\begin{proof}[Proof of \cref{lem:HoA_decision_samp}]
Assume that $\mathbf{PDQP}$ is hard on average.
Then there exists a language $L\in\mathbf{PDQP}$, a polynomial $p$, and a QPT algorithm $\cE(1^\secp)\to\bit^\secp$
such that for any QPT algorithm $\cF$ and for all sufficiently large $\secp\in\N$,
\begin{align}\label{eq:HoA_decision}
    \Pr_{x\gets\cE(1^\secp)} [\cF(x)\neq L(x)] \ge \frac{1}{p(\secp)}.
\end{align}
Because $L\in\mathbf{PDQP}$, there exists a classical deterministic polynomial-time Turing machine $\cR$ and a non-collapsing measurement oracle $\cQ$
such that for all $x\in\bit^*$,
\begin{align}\label{eq:in_PDQP_3}
    \Pr[\cR^\cQ(x) \neq L(x)] \le \negl(|x|).
\end{align}
Consider a sampling problem $\{\cR^\cQ(x)\}_{x\in\bit^*}$.
Clearly, $\{\cR^\cQ(x)\}_{x\in\bit^*}\in\mathbf{SampPDQP}$.
For the sake of contradiction, assume that $\mathbf{SampPDQP}$ is not hard on average.
Then, there exists a QPT algorithm $\cF^*$ such that for infinitely many $\secp\in\N$,
\begin{align}\label{eq:avg_in_SampPDQP}
    \SD ( \{x,\cF^*(x)\}_{x\gets\cE(1^\secp)} , \{x,\cR^\cQ(x)\}_{x\gets\cE(1^\secp)} ) \le \frac{1}{2p(\secp)}.
\end{align}
Our goal is to show that $\cF^*$ breaks \cref{eq:HoA_decision}.
By \cref{eq:in_PDQP_3,eq:avg_in_SampPDQP},
\begin{align}
    \Pr_{x\gets\cE(1^\secp)} [\cF^*(x)\neq L(x)] 
    &\le \Pr_{x\gets\cE(1^\secp)} [\cR^\cQ(x)\neq L(x)] + \underset{x\gets\cE(1^\secp)}{\mathbb{E}} \left[ \SD(\cF^*(x),\cR^\cQ(x)) \right] \\ 
    &\le \negl(\secp) + \SD( \{x,\cF^*(x)\}_{x\gets\cE(1^\secp)} , \{x,\cR^\cQ(x)\}_{x\gets\cE(1^\secp)} ) \\ 
    &\le \negl(\secp) + \frac{1}{2p(\secp)} \\ 
    &\le \frac{1}{p(\secp)},
\end{align}
holds for infinitely many $\secp\in\N$.
This contradicts \cref{eq:HoA_decision}.

\end{proof}


%% file: PDQP_OWPuzz_average.tex
\section{One-Way Puzzles from Average-Case Hardness of $\mathbf{SampPDQP}$}

In this section, we construct OWPuzzs from the average-case hardness of $\mathbf{SampPDQP}$.

\begin{theorem}\label{thm:PDQP_OWPuzz}
    If $\mathbf{SampPDQP}$ is hard on average, then OWPuzzs exist.
\end{theorem}

\begin{proof}[Proof of \cref{thm:PDQP_OWPuzz}]
Because of the equivalence of OWPuzzs and DistOWPuzzs (\cref{lem:DistOWPuzz_OWPuzz}), it suffices to construct DistOWPuzzs. 
Assume that $\mathbf{SampPDQP}$ is hard on average.
Then there exist a sampling problem $S=\{\cD_x\}_{x\in\bit^*}\in\mathbf{SampPDQP}$, a polynomial $p$, and a QPT algorithm $\cE(1^\secp)\to\bit^\secp$
such that for any QPT algorithm $\cF$ and for all sufficiently large $\secp\in\N$,

\begin{align}\label{eq:HoA_SampPDQP_condition}
    \SD (\{x,\cF(x)\}_{x\gets\cE(1^\secp)},\{x,\cD_x\}_{x\gets\cE(1^\secp)}) > \frac{1}{p(\secp)}.
\end{align}
By the definition of $\mathbf{SampPDQP}$, there exist a classical deterministic polynomial-time Turing machine $\cR$ and a non-collapsing measurement oracle $\cQ$
such that for all $x\in\bit^*$ and for all $\epsilon>0$, 
\begin{align}
    \label{eq:in_SampPDQP}
    \SD ( \cR^\cQ (x,1^{\lfloor 1/\epsilon\rfloor}) , \cD_x ) \le \epsilon.
\end{align}
In the following, we fix $\epsilon:=2p(|x|)$.
For each $x\in\bit^*$, let $C_x:=(U_1,M_1,...,U_{T_x},M_{T_x})$ be the quantum circuit that $\cR(x,1^{\lfloor 1/\epsilon\rfloor})$ queries to $\cQ$.
Let $C_x$ act on $\ell$ qubits.
Here, each $U_t$ be a unitary operator and each $M_t$ be a computational basis projective measurement on $m_i$ qubits such that $0\le m_i\le\ell$.\footnote{
    Note that all of $U_t$, $M_t$, $m_t$, and $\ell$ depend on $x$, but for the notational simplicity, we omit their dependence on $x$.
}
Let 
\begin{align}
    \label{eq:def_T_2}
    T'(\secp) := \max_{x\in\bit^\secp} \{T_x\}.
\end{align}
Note that $T'$ is a polynomial of $|x|$ because $\cR$ is a polynomial-time machine.

By using $\cR$ and $\cE$, we construct a $\frac{1}{2pT'}$-DistOWPuzz $\Samp$ as follows:
\begin{enumerate}
    \item Take $1^\secp$ as input.
    \item Sample $x\gets\cE(1^\secp)$.
    \item Let 
    $C_x:=(U_1,M_1,...,U_T,M_T)$ 
    be (a classical description of) a quantum circuit 
    that is queried to $\cQ$
    corresponding to the instance $x$.
    \item Sample $t\gets[T_x]$.
    \item Run $C_x=(U_1,M_1,...,U_t,M_t)$ on $|0^\ell\rangle$. Obtain $\tau_t:=(u_1,...,u_t)$ and the resulting state $|\psi_t^{\tau_t}\rangle$, where $u_i\in\bit^{m_i}$ is the measurement result of $M_i$ for each $i\in[t]$.
    \item Measure all qubits of $|\psi_t^{\tau_t}\rangle$ in the computational basis to obtain $v_t\in\bit^\ell$, where $v_t:=u_t\|w_t$ for some $w_t\in\bit^{\ell-m_t}$.
    \item Let $\puzz:=(x,t,\tau_t)$ and $\ans:=w_t$. Output $(\puzz,\ans)$. 
\end{enumerate}
For the sake of contradiction, we assume that $\Samp$ is not a $\frac{1}{2pT'}$-DistOWPuzz.
Then by the definition of DistOWPuzzs, there exist a QPT algorithm $\cA$ such that for infinitely many $\secp\in\N$,
\begin{align}\label{eq:break_OWPuzz}
    \SD\left( \{\puzz,\ans\}_{(\puzz,\ans)\gets\Samp(1^\secp)} , \{\puzz,\cA(1^\secp,\puzz)\}_{(\puzz,\ans)\gets\Samp(1^\secp)} \right) < \frac{1}{2p(\secp)T'(\secp)}.
\end{align}
Let $\Lambda\subseteq\N$ be the set of such $\secp$.
Note that 
\begin{align}
    \sum_{x\in\bit^\secp} \Pr[(x,t,\tau_t,w_t)\gets\Samp(1^\secp)] = \underset{x\gets\cE(1^\secp)}{\mathbb{E}} \left[ \frac{1}{T_x} \Pr[(\tau_t,w_t)\gets\cQ^t(C_x)] \right],
\end{align}
where $\cQ^t$ is the following algorithm:
\begin{enumerate}
    \item Take (a classical description of) a quantum circuit $C_x:=(U_1,M_1,...,U_{T_x},M_{T_x})$ acting on $\ell$ qubits as input.
    For each $i\in[T_x]$, $U_i$ is a unitary operator and $M_i$ is a computational basis projective measurement on $m_i$ qubits such that $0\le m_i\le\ell$.
    \item Sample $(v_1,...,v_{T_x})\gets\cQ(C_x)$, where $v_i=u_i\|w_i$ and $u_i$ is a measurement result of $M_i$ for each $i\in[T_x]$.
    \item Let $\tau_t=(u_1,...,u_t)$. Output $(\tau_t,w_t)$.
\end{enumerate} 
By \cref{eq:def_T_2,eq:break_OWPuzz}, for all $\secp\in\Lambda$,
\begin{align}
    \frac{1}{2p(\secp)T'(\secp)} 
    &> \SD\left( \{\puzz,\ans\}_{(\puzz,\ans)\gets\Samp(1^\secp)} , \{\puzz,\cA(1^\secp,\puzz)\}_{(\puzz,\ans)\gets\Samp(1^\secp)} \right) \\ 
    &= \underset{x\gets\cE(1^\secp)}{\mathbb{E}} \left[ \frac{1}{T_x} \sum_{t\in[T_x]} \SD \left( \{\tau_t,w_t\}_{(\tau_t,w_t)\gets\cQ^t(C_x)} , \{\tau_t,\cA(1^\secp,x,t,\tau_t)\}_{(\tau_t,w_t)\gets\cQ^t(C_x)} \right) \right] \\ 
    &\ge \frac{1}{T'(\secp)} \underset{x\gets\cE(1^\secp)}{\mathbb{E}} \left[ \sum_{t\in[T_x]} \SD \left( \{\tau_t,w_t\}_{(\tau_t,w_t)\gets\cQ^t(C_x)} , \{\tau_t,\cA(1^\secp,x,t,\tau_t)\}_{(\tau_t,w_t)\gets\cQ^t(C_x)} \right) \right].
\end{align}
Thus for all $\secp\in\Lambda$, 
\begin{align}
    \label{eq:avg_SD} 
    \underset{x\gets\cE(1^\secp)}{\mathbb{E}} \left[ \sum_{t\in[T_x]} \SD \left( \{\tau_t,w_t\}_{(\tau_t,w_t)\gets\cQ^t(C_x)} , \{\tau_t,\cA(1^\secp,x,t,\tau_t)\}_{(\tau_t,w_t)\gets\cQ^t(C_x)} \right) \right] < \frac{1}{2p(\secp)}.
\end{align}

Our goal is to construct a QPT algorithm $\cF$ that breaks \cref{eq:HoA_SampPDQP_condition}.
We define $\cF$ as follows:
\begin{enumerate}
    \item Take $x\in\bit^\secp$ as input.
    \item Run $\cR(x,1^{\lfloor 1/\epsilon\rfloor})$. Here instead of querying to $\cQ$, run $(v_1,...,v_{T_x})\gets\cQ^*(x,C_x)$ and use $(v_1,...,v_{T_x})$ as the outcome of $\cQ$, where $\cQ^*$ is the following QPT algorithm:
    \begin{itemize}
        \item Take $x$ and (a classical description of) a quantum circuit $C_x=(U_1,M_1,...,U_{T_x},M_{T_x})$ that acts on $\ell$ qubits as input.
        For each $t\in[T_x]$, $U_i$ is a unitary operator and $M_i$ is a computational basis projective measurement on $m_i$ qubits such that $0\le m_i\le\ell$.
        \item Run $(U_1,M_1,...,U_{T_x},M_{T_x})$ on $|0^\ell\rangle$. Obtain $(u_1,...,u_{T_x})$, where $u_i\in\bit^{m_i}$ is the measurement result of $M_i$ for each $i\in[T_x]$.
        \item For each $i\in[T_x]$, run $w_i\gets\cA(1^\secp,x,i,\tau_i)$, where $\tau_i:=(u_1,...,u_i)$. Let $v_i:=u_i\|w_i$.
        \item Output $(v_1,...,v_{T_x})$.
    \end{itemize}
\end{enumerate}

Later we will show that for all $\secp\in\Lambda$,
\begin{align}\label{eq:avg_SD_goal}
    \underset{x\gets\cE(1^\secp)}{\mathbb{E}} \left[ \SD (\cQ^*(x,C_x),\cQ(C_x)) \right] \le \frac{1}{2p(\secp)}.
\end{align}
Then by \cref{eq:in_SampPDQP,eq:avg_SD_goal},
\begin{align}
    &\SD (\{x,\cF(x)\}_{x\gets\cE(1^\secp)},\{x,\cD_x\}_{x\gets\cE(1^\secp)}) \\ 
    &\le \SD (\{x,\cR^{\cQ^*}(x,1^{\lfloor 1/\epsilon\rfloor})\}_{x\gets\cE(1^\secp)},\{x,\cR^\cQ(x,1^{\lfloor 1/\epsilon \rfloor})\}_{x\gets\cE(1^\secp)}) \\ 
    &\qquad + \SD (\{x,\cR^{\cQ}(x,1^{\lfloor 1/\epsilon\rfloor})\}_{x\gets\cE(1^\secp)},\{x,\cD_x\}_{x\gets\cE(1^\secp)}) \\ 
    &= \SD (\{x,\cR^{\cQ^*}(x,1^{\lfloor 1/\epsilon\rfloor})\}_{x\gets\cE(1^\secp)},\{x,\cR^\cQ(x,1^{\lfloor 1/\epsilon \rfloor})\}_{x\gets\cE(1^\secp)}) \\ 
    &\qquad + \underset{x\gets\cE(1^\secp)}{\mathbb{E}} \left[ \SD(\cR^\cQ(x,1^{\lfloor 1/\epsilon \rfloor}),\cD_x) \right] \\ 
    &\le \SD (\{x,\cR^{\cQ^*}(x,1^{\lfloor 1/\epsilon\rfloor})\}_{x\gets\cE(1^\secp)},\{x,\cR^\cQ(x,1^{\lfloor 1/\epsilon \rfloor})\}_{x\gets\cE(1^\secp)}) + \epsilon \\ 
    &= \underset{x\gets\cE(1^\secp)}{\mathbb{E}} \SD (\cR^{\cQ^*}(x,1^{\lfloor 1/\epsilon \rfloor}) , \cR^\cQ(x,1^{\lfloor 1/\epsilon \rfloor})) + \epsilon \\ 
    &\le \underset{x\gets\cE(1^\secp)}{\mathbb{E}} \left[ \SD (\cQ^*(x,C_x),\cQ(C_x)) \right] + \epsilon \\ 
    &\le \frac{1}{2p(\secp)} + \epsilon \\ 
    &= \frac{1}{p(\secp)}
\end{align}
for all $\secp\in\Lambda$. In the last equality, we use $\epsilon=\frac{1}{2p(\secp)}$. 
This contradicts \cref{eq:HoA_SampPDQP_condition}.

In the remaining part, we show \cref{eq:avg_SD_goal}.
To accomplish this, we define an unbounded-time algorithm $\cB$ as follows:
\begin{itemize}
    \item $\cB(k,x,C_x)\to(v_1,...,v_t,v'_{t+1},...,v'_{T_x})$:
    \begin{enumerate}
        \item Take an integer $k\in\{0,...,T_x\}$, a bit string $x\in\bit^*$, and (a classical description of) a quantum circuit $C_x=(U_1,M_1,...,U_{T_x},M_{T_x})$ that acts on $\ell$ qubits as input.
        \item Run $C_x$ on input $|0^\ell\rangle$ and let $u_t\in\bit^{m_t}$ be the outcome of the measurement $M_t$ for each $t\in[T_x]$.
        For each $t\in[T_x]$, 
        let $\tau_t:=(u_1,...,u_t)$ and let $|\psi_t^{\tau_t}\rangle$ be the (normalized) post-measurement state after the measurement $M_t$, i.e., 
        \begin{align}
            |\psi_t^{\tau_t}\rangle := \frac{ (|u_t\rangle\langle u_t|\otimes I)U_t|\psi_{t-1}^{\tau_{t-1}}\rangle }{ \sqrt{ \langle \psi_{t-1}^{\tau_{t-1}}|U_t^\dagger (|u_t\rangle\langle u_t|\otimes I)U_t|\psi_{t-1}^{\tau_{t-1}}\rangle } }.
        \end{align}
        \item For $1\le i\le k$, sample $v_i\in\bit^\ell$ with probability $|\langle v_i|\psi_i^{\tau_i}\rangle|^2$.
        \item For $k+1 \le i\le T_x$, run $w'_i\gets\cA(1^\secp,x,i,\tau_i)$ and let $v'_i:=u_i\|w'_i$.
        \item Output $(v_1,...,v_k,v'_{k+1},...,v'_{T_x})$.
    \end{enumerate}
\end{itemize}
Then, the distribution $\cB(T_x,x,C_x)$ is equivalent to the distribution $\cQ(C_x)$ and the distribution $\cB(0,x,C_x)$ is equivalent to the distribution $\cQ^*(x,C_x)$.
By the triangle inequality, 
\begin{align}
    \underset{x\gets\cE(1^\secp)}{\mathbb{E}} \left[ \SD (\cQ^*(x,C_x),\cQ(C_x)) \right] 
    &= \underset{x\gets\cE(1^\secp)}{\mathbb{E}} \left[ \SD (\cB(0,x,C_x),\cB(T_x,x,C_x)) \right] \\ 
    &\le \underset{x\gets\cE(1^\secp)}{\mathbb{E}} \left[ \sum_{t\in[T_x]} \SD (\cB(t-1,x,C_x),\cB(t,x,C_x)) \right].
\end{align}
Thus, it suffices to show that
\begin{align}
    \underset{x\gets\cE(1^\secp)}{\mathbb{E}} \left[ \sum_{t\in[T_x]} \SD (\cB(t-1,x,C_x),\cB(t,x,C_x)) \right] \le \frac{1}{2p(\secp)}
\end{align}
for all $\secp\in\Lambda$.


To show this, we define two (unbounded) algorithms as follows:
\begin{itemize}
    \item $\cQ_1(C_x,u_1,...,u_t)\to(u_{t+1},...,u_{T_x})$:
    \begin{enumerate}
        \item Take (a classical description of) a quantum circuit $C_x=(U_1,M_1,...,U_{T_x},M_{T_x})$ and bit strings $(u_1,...,u_t)$ such that $t\in[T_x]$ and $u_i\in\bit^{m_i}$ for each $i\in[T_x]$ as input.
        Here for each $i\in[T_x]$, $U_i$ is a unitary operator and $M_i$ is a computational basis projective measurement on $m_i$ qubits.
        \item Sample $(v'_1,...,v'_{T_x})\gets\cQ(C_x)$, where $v'_i=u'_i\|w'_i$ and $u'_i\in\bit^{m_i}$.
        \item If $u'_i=u_i$ for all $i\in[t]$, then output $(u_{t+1},...,u_{T_x}):=(u'_{t+1},...,u'_{T_x})$.
        Otherwise, go back to step 2.
    \end{enumerate}
    \item $\cQ_2(C_x,u_1,...,u_t)\to(w_1,...,w_t)$:
    \begin{enumerate}
        \item Take (a classical description of) a quantum circuit $C_x=(U_1,M_1,...,U_{T_x},M_{T_x})$ and bit strings $(u_1,...,u_t)$ such that $t\in[T_x]$ and $u_i\in\bit^{m_i}$ for each $i\in[T_x]$ as input.
        Here for each $i\in[T_x]$, $U_i$ is a unitary operator and $M_i$ is a computational basis projective measurement on $m_i$ qubits.
        \item Sample $(v'_1,...,v'_{T_x})\gets\cQ(C_x)$, where $v'_i=u'_i\|w'_i$ and $u'_i\in\bit^{m_i}$.
        \item If $u'_i=u_i$ for all $i\in[t]$, then output $(w_1,...,w_t):=(w'_1,...,w'_t)$.
        Otherwise, go back to step 2.
    \end{enumerate}
\end{itemize}

Then,
\begin{align}
    &\SD (\cB(t-1,x,C_x),\cB(t,x,C_x)) \\
    &= \SD ( \{v_1,...,v_{t-1},v'_t,...,v'_{T_x}\} , \{v_1,...,v_t,v'_{t+1},...,v'_{T_x}\} ) \\ 
    &= \sum_{w'_{t+1},...,w'_{T_x}} \prod_{i\in\{t+1,...,T_x\}} \Pr[w'_i\gets\cA(1^\secp,x,i,u_1,...,u_i)] \\ 
    &\qquad \times \SD ( \{v_1,...,v_{t-1},v'_t,u_{t+1}...,u_{T_x}\} , \{v_1,...,v_t,u_{t+1},...,u_{T_x}\} ) \\
    &= \SD ( \{v_1,...,v_{t-1},v'_t,u_{t+1}...,u_{T_x}\} , \{v_1,...,v_t,u_{t+1},...,u_{T_x}\} ) \\ 
    &= \sum_{u_{t+1},...,u_{T_x}} \Pr[(u_{t+1},...,u_{T_x})\gets\cQ_1(C_x,u_1,...,u_t)] \SD ( \{v_1,...,v_{t-1},v'_t\} , \{v_1,...,v_{t-1},v_t\}) \\
    &= \SD ( \{v_1,...,v_{t-1},v'_t\} , \{v_1,...,v_{t-1},v_t\}) \\ 
    &= \sum_{w_1,...,w_{t-1}} \Pr[(w_1,...,w_{t-1})\gets\cQ_2(C_x,u_1,...,u_{t-1})] \SD ( \{u_1,...,u_t,w'_t\} , \{u_1,...,u_t,w_t\} ) \\ 
    &= \SD ( \{u_1,...,u_t,w'_t\} , \{u_1,...,u_t,w_t\} ), \label{eq:henkei}
\end{align} 
where $v_i=u_i\|w_i$, $v'_i=u_i\|w'_i$, $(u_1\|w_1,...,u_{T_x}\|w_{T_x})\gets\cQ(C_x)$,  
and $w'_i\gets\cA(1^\secp,x,i,u_1,...,u_i)$.
By \cref{eq:avg_SD},
\begin{align}
    &\underset{x\gets\cE(1^\secp)}{\mathbb{E}} \left[ \sum_{t\in[T_x]} \SD (\cB(t-1,x,C_x),\cB(t,x,C_x)) \right] \\ 
    &= \underset{x\gets\cE(1^\secp)}{\mathbb{E}} \left[ \sum_{t\in[T_x]} \SD \left( \{\tau_t,w_t\}_{(\tau_t,w_t)\gets\cQ^t(C_x)} , \{\tau_t,\cA(1^\secp,x,t,\tau_t)\}_{(\tau_t,w_t)\gets\cQ^t(C_x)} \right) \right] \\ 
    &\le \frac{1}{2p(\secp)}
\end{align}
for all $\secp\in\Lambda$.

\end{proof}

%% file: adaptivePDQP.tex
\section{Adaptive PDQP and Auxiliary-Input One-Way Puzzles}
In this section we consider the adaptive queries to non-collapsing measurement oracle.
First, we define the class of decision problems that are solved with a polynomial-time deterministic algorithm that can make adaptive queries to the non-collapsing measurement oracle.
\begin{definition}[AdPDQP]
    \label{def:AdPDQP}
    A language $L$ is in $\mathbf{AdPDQP}$ if there exists a polynomial-time classical deterministic Turing machine $R$ that makes adaptive queries to a non-collapsing measurement oracle $\cQ$ such that
    \begin{itemize}
        \item For all $x\in L$, $\Pr[1\gets R^{\cQ}(x)]\ge \alpha(|x|)$,
        \item For all $x\notin L$, $\Pr[1\gets R^{\cQ}(x)]\le \beta(|x|)$,
    \end{itemize}
    where $\alpha,\beta$ are functions such that $\alpha(|x|)-\beta(|x|)\ge1/\poly(|x|)$.
\end{definition}

\begin{remark}
    As in the case of (non-adaptive) $\mathbf{PDQP}$, the error bound $(\alpha,\beta)$ in \cref{def:AdPDQP} can be amplified to $(1-\negl,\negl)$ by the repetition.
\end{remark}

Moreover, we define the class of sampling problems that are solved with a polynomial-time deterministic algorithm that can make adaptive queries to the non-collapsing measurement oracle.
\begin{definition}[SampAdPDQP]
    $\mathbf{SampAdPDQP}$ is the class of (polynomially-bounded) sampling problems $S=\{\cD_x\}_{x\in\bit^*}$ 
    for which there exists a classical deterministic polynomial-time algorithm $\cB$ that makes adaptive queries to a non-collapsing measurement oracle $\cQ$ 
    such that for all $x$ and for all $\epsilon>0$,
    $\SD(\cB^{\cQ}(x,1^{\lfloor 1/\epsilon\rfloor}),\cD_x)\le\epsilon$, 
    where $\cB^{\cQ}(x,1^{\lfloor 1/\epsilon\rfloor})$ is the output probability distribution of $\cB^{\cQ}$ on input $(x,1^{\lfloor 1/\epsilon\rfloor})$.
\end{definition}

We obtain the following lemma.
\begin{lemma}
    If $\mathbf{AdPDQP}\nsubseteq\mathbf{BQP}$, then $\mathbf{SampAdPDQP}\nsubseteq\mathbf{SampBQP}$.
\end{lemma}

Next, we show that the worst-case hardness of $\mathbf{SampAdPDQP}$ is equivalent to that of $\mathbf{SampPDQP}$.

\begin{lemma}
\label{lem:adapt_equivalence}
    $\mathbf{SampAdPDQP}\nsubseteq\mathbf{SampBQP}$ is and only if $\mathbf{SampPDQP}\nsubseteq\mathbf{SampBQP}$.
\end{lemma}
\begin{proof}[Proof of \cref{lem:adapt_equivalence}]
    The ``if'' direction holds immediately because $\mathbf{SampPDQP}\subseteq\mathbf{SampAdPDQP}$.
    
    We show the ``only if'' direction.
    Assume that $\mathbf{SampAdPDQP}\nsubseteq\mathbf{SampBQP}$.
    Then, there exists a sampling problem $\{\cD_x\}_{x\in\bit^*}$ that is contained in $\mathbf{SampAdPDQP}$ but not in $\mathbf{SampBQP}$.
    By the definition of $\mathbf{SampAdPDQP}$, there exists a classical deterministic polynomial-time algorithm $\cB$ that makes adaptive queries to $\cQ$ such that for all $x$ and all $\epsilon>0$, 
\begin{align}
    \SD (\cB^\cQ(x,1^{\lfloor 1/\epsilon\rfloor}), \cD_x)\le\epsilon.
    \label{eq:sd_D}
\end{align}
Let $N=N(x,\epsilon)$ be the number of queries that $\cB(x,1^{\lfloor 1/\epsilon\rfloor})$ makes.
    
For the sake of contradiction, assume that $\mathbf{SampPDQP}\subseteq\mathbf{SampBQP}$.
Our goal is to show that $\{\cD_x\}_{x\in\bit^*}\in\mathbf{SampBQP}$.
Let us consider the sampling problem $\{\cQ_C\}_{C\in\bit^*}$, where $\cQ_C$ is the output distribution of the following procedure: If $C$ is a classical description of some quantum circuit that has the form $(U_1,M_1,...,U_T,M_T)$, where $U_i$ is $\ell$-qubit unitary and $M_i$ is the $m_i$-qubit computational basis measurement, then query the non-collapsing measurement oracle $\cQ$ on $C$. Otherwise, output $\bot$.
Hence, we have $\{\cQ_C\}_{C\in\bit^*}\in\mathbf{SampPDQP}$ and therefore $\{\cQ_C\}_{C\in\bit^*}\in\mathbf{SampBQP}$.
This means that there exists a QPT algorithm $\cA$ such that for all $C\in\bit^*$ and for all $\epsilon>0$,
\begin{align}
    \SD (\cA(C,1^{\lfloor 1/\epsilon \rfloor}), \cQ_C) \le \epsilon. 
    \label{eq:sd_Q}
\end{align}
For each $i\in[N]$, we define the distribution $\cB_i(x,1^{\lfloor 1/\epsilon\rfloor})$ as the output distribution of the following procedure: Given $(x,1^{\floor{1/\epsilon}})$ as input, run $\cB(x,1^{\floor{1/\epsilon}})$, where the first $i$ queries are made to $\cA(\cdot,1^{\lfloor N/\epsilon\rfloor})$ and the remaining $(N-i)$ queries are made to $\cQ$.
Let $C$ be the classical description of the quantum circuit that $\cB(x,1^{\lfloor 1/\epsilon \rfloor})$ queries on its first query.
Then, for all $x$ and for all $\epsilon>0$,
\begin{align}
    \SD( \cB^\cQ(x,1^{\lfloor 1/\epsilon\rfloor}) , \cB_1(x,1^{\lfloor 1/\epsilon \rfloor}) ) 
    \le \SD(\cQ_C, \cA(C,1^{\lfloor N/\epsilon \rfloor} )
    \le \frac{\epsilon}{N}, \label{eq:sd_0}
\end{align}
where the first inequality follows from the data processing inequality, and the second from \cref{eq:sd_Q}.
Similarly, for all $i\in[N-1]$, we have 
\begin{align}
    \SD( \cB_i(x,1^{\lfloor 1/\epsilon\rfloor}) , \cB_{i+1}(x,1^{\lfloor 1/\epsilon \rfloor}) ) 
    &\le \SD( \{\cQ_{C_i}\}_{C_i\gets\cB^{\cA}(x,1^{\lfloor 1/\epsilon \rfloor})} , \{\cA(C_i,1^{\lfloor N/\epsilon\rfloor})\}_{C_i\gets\cB^\cA(x,1^{\lfloor 1/\epsilon \rfloor})} ) \\ 
    &\le \mathbb{E}_{C_i\gets\cB^\cA(x,1^{\lfloor 1/\epsilon \rfloor})} \SD( \cQ_{C_i}, \cA(C_i,1^{\lfloor N/\epsilon\rfloor})) \\ 
    &\le \frac{\epsilon}{N}, \label{eq:sd_i}
\end{align}
where $C_i$ denotes the classical description of the quantum circuit that $\cB^{\cA(\cdot,1^{\lfloor N/\epsilon \rfloor})}(x,1^{\lfloor 1/\epsilon \rfloor})$ queries on the $i$th query.
By combining \cref{eq:sd_0,eq:sd_i} and using the triangle inequality, for all $x$ and for all $\epsilon>0$,
\begin{align}
    \SD( \cB^\cQ(x,1^{\lfloor 1/\epsilon\rfloor}),\cB_T(x,1^{\lfloor 1/\epsilon\rfloor}) ) \le \epsilon. \label{eq:sd_replace}
\end{align}
Note that $\cB_T(x,1^{\lfloor 1/\epsilon \rfloor})$ corresponds to the output distribution of QPT algorithm $\cB^{\cA(\cdot,1^{\lfloor N/\epsilon\rfloor})}(x,1^{\lfloor 1/\epsilon \rfloor})$.
We consider the QPT algorithm $\cC$ that on input $(x,1^{\lfloor 1/\epsilon \rfloor})$, runs $\cB^{\cA(\cdot,1^{\lfloor 2N/\epsilon\rfloor})}(x,1^{\lfloor 2/\epsilon \rfloor})$.
Then by \cref{eq:sd_D,eq:sd_replace}, for all $x$ and for all $\epsilon>0$, 
\begin{align}
    \SD(\cC(x,1^{\lfloor 1/\epsilon \rfloor}), \cD_x)
    &\le \SD(\cC(x,1^{\lfloor 1/\epsilon \rfloor}), \cB^\cQ(x,1^{\lfloor 2/\epsilon \rfloor})) + \SD(\cD_x, \cB^\cQ(x,1^{\lfloor 2/\epsilon \rfloor})) \\ 
    &\le \frac{\epsilon}{2} + \frac{\epsilon}{2} = \epsilon.
\end{align}
Therefore, we have $\{\cD_x\}_{x\in\bit^*}\in\mathbf{SampBQP}$ and complete the proof.
\end{proof}

Next, we show that the worst-case hardness of $\mathbf{SampPDQP}$ implies auxiliary-input OWPuzzs.
\begin{theorem}
\label{thm:AI-OWPuzz}
    If $\mathbf{SampPDQP}\nsubseteq\mathbf{BQP}$, then auxiliary-input OWPuzzs exist.
\end{theorem}
\begin{proof}[Proof of \cref{thm:AI-OWPuzz}]
Because of the equivalence between auxiliary-input OWPuzzs and auxiliary-input DistOWPuzzs, it suffices to construct auxiliary-input DistOWPuzzs.
Assume that $\mathbf{SampPDQP}\nsubseteq\mathbf{SampBQP}$ and let $S=\{\cD_x\}_{x\in\bit^*}$ be a sampling problem in $\mathbf{SampPDQP}$ but not in $\mathbf{SampBQP}$.
Then, by the definition of $\mathbf{SampPDQP}$, there exist a classical deterministic polynomial-time algorithm $\cR$ that makes a single query to the non-collapsing measurement oracle $\cQ$
such that for all $x\in\bit^*$ and for all $\epsilon>0$,
\begin{align}
    \SD(\cR^{\cQ}(x,1^{\lfloor 1/\epsilon\rfloor}),\cD_x)\le\epsilon.
\end{align}
For each $x\in\bit^*$ and $\epsilon>0$, let $C_{x,\epsilon}:=(U_1,M_1,...,U_{T_{x,\epsilon}},M_{T_{x,\epsilon}})$ be the classical description of the quantum circuit that $\cR(x,1^{\floor{1/\epsilon}})$ queries to $\cQ$.
Here let $C_{x,\epsilon}$ act on $\ell$ qubits, $U_i$ be a unitary operator, and $M_i$ be a computational basis projective measurement on $m_i$ qubits such that $0\le m_i\le \ell$.
Note that all of $U_i$, $M_i$, $m_i$, and $\ell$ also depend on $x$ and $\epsilon$, but
we omit this dependence for notational simplicity.
Define 
\begin{align}\label{eq:def_K_T}
    T'(\secp) := \max_{\substack{ x\in\bit^{\le\secp}, \\ \epsilon>0:\floor{1/\epsilon}\le\secp }} \{T_{x,\epsilon}\}.
\end{align}
Here $\bit^{\le\secp}$ is a set of bit strings $x$ such that $|x|\le\secp$.
    
By using $\cR$, we construct an auxiliary-input $\frac{1}{\secp T'(\secp)}$-DistOWPuzz $\Samp$ as follows:
\begin{enumerate}
    \item Take $z\in\bit^*$ as input, where $z=(x,1^{\floor{1/\epsilon}})$ for some $x\in\bit^*$, $\epsilon>0$. 
    \item Let $C_{x,\epsilon}:=(U_1,M_1,...,U_{T_{x,\epsilon}},M_{T_{x,\epsilon}})$ be a classical description of a quantum circuit that $\cR(x,1^{\floor{1/\epsilon}})$ queries to $\cQ$.
    \item Sample $t\gets[T_{x,\epsilon}]$.
    \item Run $(U_1,M_1,...,U_t,M_t)$ on $|0^{\ell}\rangle$. Obtain $\tau_t:=(u_1,...,u_t)$ and the resulting state $|\psi^{\tau_t}_t\rangle$, where $u_i\in\bit^{m_i}$ be the measurement result of $M_i$.
    \item Measure all qubits of $|\psi^{\tau_t}_t\rangle$ in the computational basis and obtain $v_t$, where $v_t=u_i\|w_i$ for some $w_i\in\bit^{\ell-m_t}$.
    \item Let $\puzz:=(t,\tau_t)$ and $\ans:=w_t$. Output $(\puzz,\ans)$.
\end{enumerate}
For the sake of contradiction, we assume that $\Samp$ is not an auxiliary-input $\frac{1}{\secp T'(\secp)}$-DistOWPuzz.
Then by the definition of auxiliary-input DistOWPuzzs, there exists a QPT adversary $\cA$ such that for all but finitely many $z\in\bit^*$,
\begin{align}\label{eq:break_AI_OWPuzz}
    \SD\left( \{\puzz,\ans\}_{(\puzz,\ans)\gets\Samp(z)} , \{\puzz,\cA(z,\puzz)\}_{(\puzz,\ans)\gets\Samp(z)} \right) < \frac{1}{|z| T'(|z|)}.
\end{align}
Let $G:=\bit^*\setminus\Bad$ be a set of such $z\in\bit^*$, where $\Bad\subseteq\bit^*$ is a finite subset.
If $z=(x,1^{\floor{1/\epsilon}})$ for some $x\in\bit^*$, $\epsilon>0$, then
\begin{align}
    \Pr[(t,\tau_t,w_t)\gets\Samp(z)] = \frac{1}{T_{x,\epsilon}} \Pr[(\tau_t,w_t)\gets\cQ^t(C_{x,\epsilon})], 
\end{align}
where $\cQ^t$ is the following (unbounded) algorithm:
\begin{enumerate}
    \item Take (a classical description of) a quantum circuit $C_{x,\epsilon}=(U_1,M_1,...,U_{T_{x,\epsilon}},M_{T_{x,\epsilon}})$. 
    \item Sample $(v_1,...,v_{T_{x,\epsilon}})\gets\cQ(C_{x,\epsilon})$, where $v_i=u_i\|w_i$ and $u_i$ is a measurement result of $M_i$ for each $i\in[T_{x,\epsilon}]$.
    \item Let $\tau_t:=(u_1,...,u_t)$. Output $(\tau_t,w_t)$.
\end{enumerate}
    
By \cref{eq:break_AI_OWPuzz}, for all $z\in G$ such that $z=(x,1^{\floor{1/\epsilon}})$,
\begin{align}
    \frac{1}{|z| T'(|z|)}
    &> \SD\left( \{\puzz,\ans\}_{(\puzz,\ans)\gets\Samp(z)} , \{\puzz,\cA(z,\puzz)\}_{(\puzz,\ans)\gets\Samp(z)} \right) \\ 
    &= \frac{1}{T_{x,\epsilon}} \sum_{t\in[T_{x,\epsilon}]} \SD \left( \{\tau_t,w_t\} , \{\tau_t,\cA(z,t,\tau_t)\} \right),
\end{align}
where $(\tau_t,w_t)\gets\cQ^t(C_{x,\epsilon})$.
Thus by \cref{eq:def_K_T}, for all $z\in G$ such that $z=(x,1^{\floor{1/\epsilon}})$,
\begin{align}\label{eq:dist}
    \sum_{t\in[T_{x,\epsilon}]} \SD\left( \{\tau_t,w_t\}, \{\tau_t,\cA(z,t,\tau_t)\} \right) < \epsilon,
\end{align}
where $(\tau_t,w_t)\gets\cQ^t(C_{x,\epsilon})$.
    
Our goal is to construct a QPT algorithm $\cF$ on input $(x,1^{\floor{1/\epsilon}})$ 
such that for all $x\in\bit^*$ and for all $\epsilon>0$
\begin{align}\label{eq:final_goal}
    \SD (\cF(x,1^{\floor{1/\epsilon}}),\cD_x) \le \epsilon.
\end{align}
We define $\cF$ as follows:
\begin{enumerate}
    \item Take $x\in\bit^*$ and $1^{\floor{1/\epsilon}}$ as input. 
    \item Let $b:=\max\{|z|:z\in\Bad\}$. If $|x|+\floor{1/\epsilon} >b$, then do the following:
    Run $\cR(x,1^{\floor{2/\epsilon}})$. Let $C_{x,\epsilon/2}:=(U_1,M_1,...,U_{T_{x,\epsilon/2}},M_{T_{x,\epsilon/2}})$ be the query that $\cR$ makes to $\cQ$.
    Instead of the query to $\cQ$, run the following QPT algorithm $V\gets\cQ^*(x,1^{\floor{2/\epsilon}},C_{x,\epsilon/2})$ and use $V$ as the outcome of $\cQ$:
    \begin{enumerate}
        \item Take $x$, $1^{\floor{2/\epsilon}}$, and (a classical description of) a quantum circuit $C_{x,\epsilon/2}=(U_1,M_1,...,U_{T_{x,2/\epsilon}},M_{T_{x,2/\epsilon}})$ that acts on $\ell$ qubits as input.
        \item Run $C_{x,\epsilon/2}$ on $|0^{\ell}\rangle$. Obtain $(u_1,...,u_{T_{x,\epsilon/2}})$, where $u_i\in\bit^{m_i}$ is the measurement outcome of $M_i$. 
        \item For each $i\in[T_{x,\epsilon/2}]$, run $w_i\gets\cA((x,1^{\floor{2/\epsilon}}),i,u_1,...,u_i)$.
        Let $V:=(u_1\|w_i,...,u_{T_{x,\epsilon/2}}\|w_{T_{x,\epsilon/2}})$.
        \item Output $V$.
    \end{enumerate}
    \item If $|x|+\floor{1/\epsilon}\le b$, do the following: 
    Run $\cR(x,1^{\floor{1/\epsilon}})$. Let $C_{x,\epsilon}:=(U_1,M_1,...,U_{T_{x,\epsilon}},M_{T_{x,\epsilon}})$ be the query that $\cR$ makes to $\cQ$. 
    Instead of the query to $\cQ$, run the following algorithm: 
    \begin{enumerate}
        \item Take (a classical description of) a quantum circuit $C_{x,\epsilon}=(U_1,M_1,...,U_{T_{x,\epsilon}},M_{T_{x,\epsilon}})$ that acts on $\ell$ qubits as input.
        \item For each $i\in[T_{x,\epsilon}]$ do the following:
        \begin{enumerate}
            \item Run $(U_1,M_1,...,U_i,M_i)$ on $|0^\ell\rangle$. Obtain $\tau_i:=(u'_1,...,u'_i)$, where $u_t\in\bit^{m_t}$ is the measurement outcome of $M_t$. Let $|\psi^{\tau_i}_i\rangle$ be the resulting state.
            \item If $(u'_1,...,u'_{i-1})=(u_1,...,u_{i-1})$, then proceed to the next step. Otherwise, go back to the previous step.
            \item Let $u_i=u'_i$. Measure all qubits of $|\psi^{\tau_i}_i\rangle$ in the computational basis and obtain $v_i\in\bit^\ell$, where $v_i=u_i\|w_i$ for some $w_i\in\bit^{\ell-m_i}$.
        \end{enumerate}
        \item Output $(v_1,...,v_{T_{x,\epsilon}})$.
    \end{enumerate}
\end{enumerate}
    
Note that if $x$ and $1^{\floor{1/\epsilon}}$ satisfies $|x|+\floor{1/\epsilon}\le b$, then $\cF(x,1^{\floor{1/\epsilon}})$ runs in constant time.
This is because $\Bad$ is the finite set and therefore $b$ is constant that does not depend on $|x|$ and $\floor{1/\epsilon}$.
Moreover, if $|x|+\floor{1/\epsilon}\le b$, then the output distribution of $\cF(x,1^{\floor{1/\epsilon}})$ is equivalent to $\cR^\cQ(x,1^{\floor{1/\epsilon}})$.
Thus, for all $x\in\bit^*$ and $\epsilon>0$ that satisfies $|x|+\floor{1/\epsilon}\le b$, 
\begin{align}
    \SD (\cF(x,1^{\floor{1/\epsilon}}),\cD_x) 
    &= \SD (\cR^\cQ(x,1^{\floor{1/\epsilon}}),\cD_x) \le \epsilon.
\end{align}
    
Next, we consider the case where $|x|+\floor{1/\epsilon}>b$.
Later we will show that 
\begin{align}
    \label{eq:worst_goal}
    \SD ( \cQ^*(x,1^{\floor{1/\epsilon}},C_{x,\epsilon}) , \cQ(C_{x,\epsilon}) ) \le \epsilon
\end{align}
for all $x\in\bit^*$ and $\epsilon>0$ such that $|x|+\floor{1/\epsilon}>b$.
Then,
\begin{align}
    \SD (\cF(x,1^{\floor{1/\epsilon}}) , \cD_x) 
    &\le \SD (\cR^{\cQ^*}(x,1^{\floor{2/\epsilon}}) , \cR^\cQ(x,1^{\floor{2/\epsilon}}) ) + \frac{\epsilon}{2} \\ 
    &\le \SD ( \cQ^*(x,1^{\floor{2/\epsilon}},C_{x,\epsilon/2}) , \cQ(C_{x,\epsilon/2}) ) + \frac{\epsilon}{2} \\ 
    &\le \frac{\epsilon}{2} + \frac{\epsilon}{2} \le \epsilon.
\end{align}
Therefore, $\cF$ satisfies \cref{eq:final_goal}.
    
In the remaining part, we show \cref{eq:worst_goal}.
To accomplish this, we define the following (unbounded) algorithm $\cB$:
\begin{itemize}
    \item $\cB(k,z,C_{x,\epsilon})\to(v_1,...,v_k,v'_{k+1},...,v'_{T_{x,\epsilon}})$:
    \begin{enumerate}
        \item Take $k\in\{0,...,T_{x,\epsilon}\}$, $z=(x,1^{\floor{1/\epsilon}})$, and (a classical description of) a quantum circuit $C_{x,\epsilon}=(U_1,M_1,...,U_{T_{x,\epsilon}},M_{T_{x,\epsilon}})$ that acts on $\ell$ qubits as input.
        \item Run $C_{x,\epsilon}$ on input $|0^{\ell}\rangle$ and let $u_t\in\bit^{m_t}$ be the outcome of the measurement $M_t$ for each $t\in[T_{x,\epsilon}]$.
        For each $t\in[T_{x,\epsilon}]$, 
        let $\tau_t:=(u_1,...,u_t)$ and let $|\psi_t^{\tau_t}\rangle$ be the (normalized) post-measurement state after the measurement $M_t$, i.e., 
        \begin{align}
            |\psi_t^{\tau_t}\rangle := \frac{ (|u_t\rangle\langle u_t|\otimes I)U_t|\psi_{t-1}^{\tau_{t-1}}\rangle }{ \sqrt{ \langle \psi_{t-1}^{\tau_{t-1}}|U_t^{\dagger} (|u_t\rangle\langle u_t|\otimes I)U_t|\psi_{t-1}^{\tau_{t-1}}\rangle } }.
        \end{align}
        \item For $1\le i\le k$, sample $v_i\in\bit^\ell$ with probability $|\langle v_i|\psi_i^{\tau_i}\rangle|^2$.
        \item For $k+1 \le i\le T_{x,\epsilon}$, run $w'_i\gets\cA(z,i,\tau_i)$ and let $v'_i:=u_i\|w'_i$.
        \item Output $(v_1,...,v_k,v'_{k+1},...,v'_{T_{x,\epsilon}})$.
    \end{enumerate}
\end{itemize}
Then, the distribution $\cB(T_{x,\epsilon},z,C_{x,\epsilon})$ is equivalent to the distribution $\cQ(C_{x,\epsilon})$ and the distribution $\cB(0,z,C_{x,\epsilon})$ is equivalent to the distribution $\cQ^*(x,1^{\floor{1/\epsilon}},C_{x,\epsilon})$.
By the triangle inequality, it suffices to show 
\begin{align}
    \sum_{t\in[T_{x,\epsilon}]} \SD (\cB(t-1,z,C_{x,\epsilon}),\cB(t,z,C_{x,\epsilon})) \le \epsilon
\end{align}
for all $x\in\bit^*$ and $\epsilon>0$ such that $|x|+\floor{1/\epsilon}>b$.
Indeed, for all $x\in\bit^*$, $\epsilon>0$, and for all $t\in[T_{x,\epsilon}]$,
\begin{align}
    &\SD (\cB(t-1,z,C_{x,\epsilon}),\cB(t,z,C_{x,\epsilon})) \\ 
    &= \SD\left( \{\tau_t,w_t\}_{(\tau_t,w_t)\gets\cQ^t(C_{x,\epsilon})}, \{\tau_t,\cA(z,t,\tau_t)\}_{(\tau_t,w_t)\gets\cQ^t(C_{x,\epsilon})} \right),
\end{align}
where $z=(x,1^{\floor{1/\epsilon}})$.
We can obtain the above equality by the same way as \cref{eq:henkei}.
If $|x|+\floor{1/\epsilon}>b$, then $|z|>b$ and therefore $z\in G$.
By \cref{eq:dist}, for all $x\in\bit^*$ and  $\epsilon>0$ such that $|x|+\floor{1/\epsilon}>b$,
\begin{align}
    \sum_{t\in[T_{x,\epsilon}]} \SD (\cB(t-1,z,C_{x,\epsilon}),\cB(t,z,C_{x,\epsilon})) \le \epsilon.
\end{align}
\end{proof}

By combining \cref{lem:adapt_equivalence,thm:AI-OWPuzz}, we obtain the following corollary.
\begin{corollary}
    If $\mathbf{SampAdPDQP}\nsubseteq\mathbf{BQP}$, then auxiliary-input OWPuzzs exist.
\end{corollary}

%% file: CROWPuzz.tex
\section{Distributional Collision-Resistant Puzzles}
\label{sec:CROWPuzz}
In this section, we introduce a quantum analogue of dCRH, namely,
distributional collision-resistant puzzles (dCRPuzzs).

\subsection{Definition of classical dCRH}
Before introducing the quantum analogue,
we first remind the definition of classical dCRH for the convenience of readers.

\begin{definition}[Distributional Collision-Resistant Hashing (dCRH)~\cite{STOC:DubIsh06,EC:BHKY19}]
Let $\{\cH_\secp : \bit^{n(\secp)} \to \bit^{m(\secp)}\}_{\secp\in\mathbb{N}}$
be an efficient function family ensemble. Here $n$ and $m$ are polynomials.
We say that it is a distributional collision-resistant hash (dCRH) function family 
if there exists a polynomial $p$ such that for any QPT algorithm $\cA$, 
\begin{align}
\SD(\{h, \cA(h)\}_{h\gets \cH_\secp}, \{h, \mathsf{Col}(h)\}_{h\gets \cH_\secp}) \ge\frac{1}{p(\secp)}
\end{align}
for all sufficiently large $\secp\in\mathbb{N}$.
Here $\mathsf{Col}(h)$ is the following distribution.
\begin{enumerate}
    \item 
    Sample $x\gets\bit^{n(\secp)}$.
    \item 
    Sample $x'\gets h^{-1}(h(x))$.
    \item 
    Output $(x,x')$.
\end{enumerate}
\end{definition}

dCRH imply distributional OWFs~\cite{EC:BHKY19}, and therefore OWFs.
Average-case hardness of $\mathbf{SZK}$ imply dCRH~\cite{C:KomYog18,EC:BHKY19}.
dCRH will not be constructed from one-way permutations (OWPs) in a black-box way~\cite{EC:Simon98,STOC:DubIsh06}.
dCRH will not be constructed from iO plus OWPs in a black-box way~\cite{SIAM:AshSeg16}.
dCRH implies constant-round statistically-hiding commitments~\cite{EC:BHKY19}.
Two-message statistically-hiding commitments imply dCRH~\cite{EC:BHKY19}.

\subsection{Definition of dCRPuzzs}
Next we introduce dCRPuzzs.

\begin{definition}[Distributional Collision-Resistant Puzzles (dCRPuzzs)]
A distributional collision-resistant puzzle (dCRPuzz)
is a pair $(\Setup,\Samp)$ of algorithms such that
\begin{itemize}
\item 
$\Setup(1^\secp)\to\pp:$
It is a QPT algorithm that, on input the security parameter $\secp$,
outputs a classical public parameter $\pp$.
\item 
$\Samp(\pp)\to(\puzz,\ans):$
It is a QPT algorithm that, on input $\pp$, outputs two bit strings $(\puzz,\ans)$. 
\end{itemize}
We require the following property:
there exists a polynomial $p$ such that for any QPT adversary $\cA$ 
\begin{align}
\SD(\{\pp,\cA(\pp)\}_{\pp\gets\Setup(1^\secp)},\{\pp,\mathsf{Col}(\pp)\}_{\pp\gets\Setup(1^\secp)} )\ge\frac{1}{p(\secp)}   
\end{align}
for all sufficiently large $\secp\in\mathbb{N}$, where $\mathsf{Col}(\pp)$ is the following distribution:
\begin{enumerate}
    \item 
    Run $(\puzz,\ans)\gets\Samp(\pp)$.
    \item 
    Sample $\ans'$ with the conditional probability $\Pr[\ans'|\puzz]\coloneqq\frac{\Pr[(\ans',\puzz)\gets\Samp(\pp)]}{\Pr[\puzz\gets\Samp(\pp)]}$. 
    \item 
    Output $(\puzz,\ans,\ans')$.
\end{enumerate}
\end{definition}

The following lemma is easy to show.
\begin{lemma}
If (quantumly-secure) dCRH exists, then dCRPuzzs exist.  
\end{lemma}

\subsection{dCRPuzzs imply average-case hardness of $\mathbf{SampPDQP}$}
\begin{theorem}
If dCRPuzzs exist, then $\mathbf{SampPDQP}$ is hard on average.
\end{theorem}
\begin{proof}
Let $(\Setup,\Samp)$ be a dCRPuzz.
Without loss of generality, we can assume that $\Samp(\pp)\to(\puzz,\ans)$ runs as follows.
\begin{enumerate}
    \item 
    Apply a unitary $V_\pp$ on $|0....0\rangle$ to generate a state, 
    \begin{align}
    V_\pp|0...0\rangle=\sum_{\puzz,\ans}c_{\puzz,\ans}|\puzz\rangle_\regA|\ans\rangle_\regB|junk_{\puzz,\ans}\rangle,
    \end{align}
    where $c_{\puzz,\ans}$ is a complex coefficient, and $|junk_{\puzz,\ans}\rangle$ is a ``junk'' state.
    \item 
    Measure the register $\regA$ to get $\puzz$.
    \item 
    Measure the register $\regB$ to get $\ans$.
    \item 
    Output $(\puzz,\ans)$.
\end{enumerate}
Define the following distribution $\cD_{\pp}$.
    \begin{enumerate}
        \item 
        Apply $V_\pp$ on $|0...0\rangle$ to generate 
        \begin{align}
        V_\pp|0...0\rangle=\sum_{\puzz,\ans} c_{\puzz,\ans} |\puzz\rangle_\regA|\ans\rangle_\regB|junk_{\puzz,\ans}\rangle_\regC.
        \end{align}
        \item 
        Measure the register $\regA$. The state is collapsed to
        \begin{align}
        |\puzz\rangle\otimes\left(\sum_{\ans} c_{\puzz,\ans} |\ans\rangle_\regB|junk_{\puzz,\ans}\rangle_\regC\right) 
        \end{align}
        up to the normalization.
        \item 
        Perform the non-collapsing measurement on this state to sample $v_1\coloneqq(\puzz,\ans,junk)$.
        \item 
        Perform the non-collapsing measurement on this state to sample $v_2\coloneqq(\puzz,\ans',junk')$.
        \item 
        Output $(\puzz,\ans,\ans')$.
    \end{enumerate}
    It is clear that the sampling problem $\{\cD_{\pp}\}_{\pp}$ is in $\mathbf{SampPDQP}$.
    (The classical polynomial-time deterministic base algorithm has only to query $(U_1,M_1,U_2,M_2)$ to the non-collapsing measurement oracle,
    where $U_1=V_\pp$, $M_1$ is the measurement after the application of $V_\pp$, $U_2$ is the identity, and $M_2$ does not do any measurement.)
    
Assume that $\mathbf{SampPDQP}$ is not hard on average.
Then, from the definition of average-case hardness of $\mathbf{SampPDQP}$ (\cref{def:AHSampPDQP}), we have that for any polynomial $p$ 
    there exists a QPT algorithm $\cF$ such that
    \begin{align}
        \SD ( \{\pp,\cF(\pp)\}_{\pp\gets\Setup(1^\secp)}, \{\pp,\cD_\pp\}_{\pp\gets\Setup(1^\secp)} ) \le \frac{1}{p(\secp)}
    \end{align}
    for infinitely-many $\secp\in\mathbb{N}$.

From such $\cF$, we construct a QPT adversary $\cA$ that breaks the dCRPuzz as follows.
\begin{enumerate}
\item 
Receive $\pp$ as input. 
   \item 
    Run $(\puzz,\ans,\ans')\gets\cF(\pp)$.
    \item 
    Output $(\puzz,\ans,\ans')$.
\end{enumerate}
Then,
\begin{align}
&\SD(\{\pp,\cA(\pp)\}_{\pp\gets\Setup(1^\secp)},\{\pp,\mathsf{Col}(\pp)\}_{\pp\gets\Setup(1^\secp)})    \\
&=
\SD(\{\pp,\cF(\pp)\}_{\pp\gets\Setup(1^\secp)},\{\pp,\mathsf{Col}(\pp)\}_{\pp\gets\Setup(1^\secp)})   \\
&=
\SD(\{\pp,\cF(\pp)\}_{\pp\gets\Setup(1^\secp)},\{\pp,\cD_\pp\}_{\pp\gets\Setup(1^\secp)})   \\
&\le
\frac{1}{p(\secp)}
\end{align}
for infinitely-many $\secp\in\mathbb{N}$, which
means that $\cA$ breaks the dCRPuzz, but it is the contradiction.
\end{proof}

%% file: oneshot.tex
\section{One-Shot Signatures and MACs}
\label{sec:money}
\subsection{Definitions}
We first remind the definition of one-shot signatures.
\begin{definition}[One-Shot Signatures~\cite{STOC:AGKZ20}]\label{def:OSSDS}
A one-shot signature scheme
is a set $(\Setup,\Gen,\Sign,\Ver)$ of algorithms such that
\begin{itemize}
\item 
$\Setup(1^\secp)\to \pp:$
It is a QPT algorithm that, on input the security parameter $\secp$, outputs a public parameter $\pp$.
    \item 
    $\Gen(\pp)\to (\vk,\sigk):$
    It is a QPT algorithm that, on input $\pp$, outputs
    a quantum signing key $\sigk$ and a classical verification key $\vk$.
    \item 
    $\Sign(\sigk,m)\to\sigma:$
    It is a QPT algorithm that, on input $\sigk$ and a message $m$, outputs a classical signature $\sigma$.
    \item 
    $\Ver(\pp,\vk,\sigma,m)\to\top/\bot:$
    It is a QPT algorithm that, on input $\pp$, $\vk$, $\sigma$, and $m$, outputs $\top/\bot$.
\end{itemize}
We require the following properties.

\paragraph{Correctness:}
For any $m$,
\begin{align}
\Pr\left[
\top\gets\Ver(\pp,\vk,\sigma,m):
\begin{array}{rr}
\pp\gets\Setup(1^\secp)\\
(\vk,\sigk)\gets\Gen(\pp)\\
\sigma\gets\Sign(\sigk,m)
\end{array}
\right]\ge1-\negl(\secp).
\end{align}

\paragraph{Security:}
For any QPT adversary $\cA$,
\begin{align}
\Pr\left[
\begin{array}{c}
m_0\neq m_1\\
\wedge\\
\top\gets\Ver(\pp,\vk,\sigma_0,m_0)\\
\wedge\\
\top\gets\Ver(\pp,\vk,\sigma_1,m_1)\\
\end{array}
:
\begin{array}{rr}
\pp\gets\Setup(1^\secp)\\
(\vk,m_0,m_1,\sigma_0,\sigma_1)\gets\cA(\pp)
\end{array}
\right]\le\negl(\secp).
\end{align}
\end{definition}

One-shot MACs are a relaxation of one-shot signatures and two-tier one-shot signatures~\cite{cryptoeprint:2023/1937}, which have partial public verification.
One-shot MACs can be constructed from the LWE assumption~\cite{cryptoeprint:2023/1937,untele}.
\begin{definition}[One-Shot MACs~\cite{untele}]\label{def:OSSMAC}
A one-shot message authentication code (MAC) scheme
is a set $(\Setup,\Gen,\Sign,\Ver)$ of algorithms such that
\begin{itemize}
\item 
$\Setup(1^\secp)\to (\pp,\mvk):$
It is a QPT algorithm that, on input the security parameter $\secp$, outputs a public parameter $\pp$ and a master verification key $\mvk$.
    \item 
    $\Gen(\pp)\to (\vk,\sigk):$
    It is a QPT algorithm that, on input $\pp$, outputs
    a quantum signing key $\sigk$ and a classical verification key $\vk$.
    \item 
    $\Sign(\sigk,m)\to\sigma:$
    It is a QPT algorithm that, on input $\sigk$ and a message $m$, outputs a classical signature $\sigma$.
    \item 
    $\Ver(\pp,\mvk,\vk,\sigma,m)\to\top/\bot:$
    It is a QPT algorithm that, on input $\pp$, $\mvk$, $\vk$, $\sigma$, and $m$, outputs $\top/\bot$.
\end{itemize}
We require the following properties.

\paragraph{Correctness:}
For any $m$,
\begin{align}
\Pr\left[
\top\gets\Ver(\pp,\mvk,\vk,\sigma,m):
\begin{array}{rr}
(\pp,\mvk)\gets\Setup(1^\secp)\\
(\sigk,\vk)\gets\Gen(\pp)\\
\sigma\gets\Sign(\sigk,m)
\end{array}
\right]\ge1-\negl(\secp).
\end{align}

\paragraph{Security:}
For any QPT adversary $\cA$,
\begin{align}
\Pr\left[
\begin{array}{c}
m_0\neq m_1\\
\wedge\\
\top\gets\Ver(\pp,\mvk,\vk,\sigma_0,m_0)\\
\wedge\\
\top\gets\Ver(\pp,\mvk,\vk,\sigma_1,m_1)\\
\end{array}
:
\begin{array}{rr}
(\pp,\mvk)\gets\Setup(1^\secp)\\
(\vk,m_0,m_1,\sigma_0,\sigma_1)\gets\cA(\pp)
\end{array}
\right]\le\negl(\secp).
\end{align}
\end{definition}

\subsection{One-shot MACs imply dCRPuzzs}
\begin{theorem}\label{thm:MAC_to_dCRP}
If one-shot MACs exist, then dCRPuzzs exist.
\end{theorem}

\begin{proof}
Let $(\Setup,\Gen,\Sign,\Ver)$ be a one-shot MAC.
Define the algorithm $\cC$ as follows:
\begin{enumerate}
    \item 
    Get $\pp$ as input.
    \item 
    Run $(\vk,\sigk)\gets\Gen(\pp)$.
    \item 
    Choose $m_0\gets\bit^\ell$.
    Run $\sigma_0\gets\Sign(\sigk,m_0)$.
    \item 
    Run $\Gen(\pp)$ until $\vk$ is obtained.
    \item
    Choose $m_1\gets\bit^\ell$.
    Run $\sigma_1\gets\Sign(\sigk,m_1)$.
    \item 
    Output $(\vk,m_0,\sigma_0,m_1,\sigma_1)$.
\end{enumerate}

Let $\Pi$ be a POVM element corresponding to the event
that the challenger of the security game of one-shot MACs accepts.
From the correctness of the one-shot MAC, we have
\begin{align}
\sum_{\pp,\mvk}\Pr[(\pp,\mvk)\gets\Setup(1^\secp)]    
\Tr[\Pi(|\mvk\rangle\langle\mvk|\otimes|\pp\rangle\langle\pp|\otimes \cC(\pp))]
\ge1-\negl(\secp).
\end{align}
Let $q$ be a polynomial. Define the set $G$ as
\begin{align}
G\coloneqq\left\{(\pp,\mvk):
\Tr[\Pi(|\mvk\rangle\langle\mvk|\otimes|\pp\rangle\langle\pp|\otimes \cC(\pp))]
\ge1-\frac{1}{q(\secp)}
\right\}.
\end{align}
Then, from the standard average argument,
\begin{align}
\sum_{(\pp,\mvk)\in G}\Pr[(\pp,\mvk)\gets\Setup(1^\secp)]    
\ge1-\negl(\secp).
\end{align}

We construct a dCRPuzz $(\mathsf{d}.\Setup,\mathsf{d}.\Samp)$
as follows.
\begin{itemize}
    \item 
    $\mathsf{d}.\Setup(1^\secp)\to\mathsf{d}.\pp:$
    Run $(\pp,\mvk)\gets\Setup(1^\secp)$.
    Output $\mathsf{d}.\pp\coloneqq \pp$.
    \item 
    $\mathsf{d}.\Samp(\mathsf{d}.\pp)\to(\puzz,\ans):$
    Parse $\mathsf{d}.\pp=\pp$.
    Run $(\vk,\sigk)\gets\Gen(\pp)$.
    Choose $m\gets\bit^\ell$.
    Run $\sigma\gets\Sign(\sigk,m)$.
    Output $\puzz\coloneqq \vk$ and $\ans\coloneqq(m,\sigma)$.
\end{itemize}
For the sake of contradiction, assume that this is not a dCRPuzz.
Then, for any polynomial $p$, there exists a QPT adversary $\cA$ such that
\begin{align}
\SD((\mathsf{d}.\pp, \cA(\mathsf{d}.\pp))_{\mathsf{d}.\pp\gets \mathsf{d}.\Setup(1^\secp)}, (\mathsf{d}.\pp, \mathsf{Col}(\mathsf{d}.\pp))_{\mathsf{d}.\pp\gets \mathsf{d}.\Setup(1^\secp)}) \le\frac{1}{p(\secp)}
\label{OSSeq}
\end{align}
for infinitely many $\secp\in\mathbb{N}$.
Here $\mathsf{Col}(\mathsf{d}.\pp)$ is the following distribution.
\begin{enumerate}
    \item 
    Run $(\puzz,\ans)\gets\mathsf{d}.\Samp(\mathsf{d}.\pp)$.
    \item 
    Sample $\ans'$ with the conditional probability $\Pr[\ans'|\puzz]\coloneqq\frac{\Pr[(\ans',\puzz)\gets\mathsf{d}.\Samp(\mathsf{d}.\pp)]}{\Pr[\puzz\gets\mathsf{d}.\Samp(\mathsf{d}.\pp)]}$. 
    \item 
    Output $(\puzz,\ans,\ans')$.
\end{enumerate}

From \cref{OSSeq}, we have
\begin{align}
\frac{1}{p(\secp)}
&\ge
\SD((\mathsf{d}.\pp, \cA(\mathsf{d}.\pp))_{\mathsf{d}.\pp\gets \mathsf{d}.\Setup(1^\secp)}, (\mathsf{d}.\pp, \mathsf{Col}(\mathsf{d}.\pp))_{\mathsf{d}.\pp\gets \mathsf{d}.\Setup(1^\secp)}) \\
&\ge
\sum_{\mathsf{d}.\pp}\Pr[\mathsf{d}.\pp\gets\mathsf{d}.\Setup(1^\secp)]
\TD(|\mathsf{d}.\pp\rangle\langle \mathsf{d}.\pp|\otimes\cA(\mathsf{d}.\pp),
|\mathsf{d}.\pp\rangle\langle\mathsf{d}.\pp|\otimes \mathsf{Col}(\mathsf{d}.\pp))\\
&=
\sum_{\pp}\Pr[\pp\gets\Setup(1^\secp)]
\TD(|\pp\rangle\langle \pp|\otimes\cA(\pp),
|\pp\rangle\langle\pp|\otimes \cC(\pp)).
\end{align}
If we define the set $S$ as
\begin{align}
S\coloneqq\left\{
\pp:
\TD(|\pp\rangle\langle \pp|\otimes\cA(\pp),|\pp\rangle\langle\pp|\otimes \cC(\pp))\le\frac{1}{\sqrt{p(\secp)}}
\right\},    
\end{align}
we have
\begin{align}
\sum_{\pp\in S}\Pr[\pp\gets\Setup(1^\secp)]\ge1-\frac{1}{\sqrt{p(\secp)}}
\end{align}
from the standard average argument.

From $\cA$, we can construct a QPT adversary $\cB$ that breaks the security of the one-shot MAC as follows.
\begin{enumerate}
    \item 
    Receive $\pp$ as input.
    \item 
    Run $(\vk,m_0,\sigma_0,m_1,\sigma_1)\gets\cA(\pp)$.
    \item 
    Output
   $(\vk,m_0,\sigma_0,m_1,\sigma_1)$.
\end{enumerate}
The probability that $\cB$ wins is
\begin{align}
&\sum_{\pp,\mvk}\Pr[(\pp,\mvk)\gets\Setup(1^\secp)]
\mbox{Tr}[\Pi(|\mvk\rangle\langle\mvk|\otimes|\pp\rangle\langle\pp|\otimes\cB(\pp))]    \\
&\ge\sum_{(\pp,\mvk)\in G\wedge \pp\in S}\Pr[(\pp,\mvk)\gets\Setup(1^\secp)]\mbox{Tr}[\Pi(|\mvk\rangle\langle\mvk|\otimes|\pp\rangle\langle\pp|\otimes\cB(\pp))]    \\
&=\sum_{(\pp,\mvk)\in G\wedge \pp\in S}\Pr[(\pp,\mvk)\gets\Setup(1^\secp)]\mbox{Tr}[\Pi(|\mvk\rangle\langle\mvk|\otimes|\pp\rangle\langle\pp|\otimes\cA(\pp))]    \\
&\ge\sum_{(\pp,\mvk)\in G\wedge \pp\in S}\Pr[(\pp,\mvk)\gets\Setup(1^\secp)]\mbox{Tr}[\Pi(|\mvk\rangle\langle\mvk|\otimes|\pp\rangle\langle\pp|\otimes\cC(\pp))]    
-\frac{1}{\sqrt{p(\secp)}}\\
&\ge \left(1-\frac{1}{\sqrt{p(\secp)}}\right)\left(1-\frac{1}{q(\secp)}\right)-\frac{1}{\sqrt{p(\secp)}}
\end{align}
for infinitely many $\secp\in\mathbb{N}$.

\end{proof}

%% file: commitments.tex
\section{Commitments}

\subsection{Definitions}
We first remind the definition of commitments we consider.
\begin{definition}[Two-Message Honest-Statistically-Hiding Computationally-Binding Bit Commitments with Classical Communication]
A two-message honest-statistically-hiding and computationally-binding bit commitment scheme with classical communication is a set 
$(\mathsf{S}_1,\mathsf{S}_2,\mathsf{R}_1,\mathsf{R}_2)$ of algorithms such that
\begin{enumerate}
    \item 
    $\mathsf{R}_1(1^\secp)\to (r_1,\psi_R):$
    It is a QPT algorithm that, on input the security parameter $\secp$,
    outputs a classical bit string $r_1$ and an internal quantum state $\psi_R$.
    \item 
    $\mathsf{S}_1(r_1,b)\to (s_1,\psi_S):$
    It is a QPT algorithm that, on input $r_1$ and a bit $b\in\bit$, outputs a bit string $s_1$ and an internal state $\psi_S$.
    \item 
    $\mathsf{S}_2(b,\psi_S)\to s_2:$
    It is a QPT algorithm that, on input $b$ and $\psi_S$, outputs a bit string $s_2$.
    \item 
    $\mathsf{R}_2(\psi_R,s_1,s_2,b)\to \top/\bot:$
    It is a QPT algorithm that, on input $\psi_R$, $s_1$, $s_2$, and $b$, outputs $\top/\bot$.
\end{enumerate}
We require the following properties.
\paragraph{Correctness.}
For all $b\in\bit$,
\begin{align}
\Pr[\top\gets\mathsf{R}_2(\psi_R,s_1,s_2,b):(r_1,\psi_R)\gets\mathsf{R}_1(1^\secp),(s_1,\psi_S)\gets\mathsf{S}_1(r_1,b),s_2\gets\mathsf{S}_2(b,\psi_S)]    
\ge1-\negl(\secp).
\end{align}
\paragraph{Honest statistical hiding.}
For all $b\in\bit$ and for any (not-necessarily-efficient) algorithm $\cA$,
\begin{align}
\Pr[b\gets\cA(\psi_R,s_1):(r_1,\psi_R)\gets\cR_1(1^\secp),s_1\gets\mathsf{S}_1(r_1,b)]    
\le\frac{1}{2}+
\negl(\secp).
\end{align}
\paragraph{Computational binding.}
For any QPT algorithm $\cA$,
\begin{align}
\Pr\left[
\begin{array}{c}
\top\gets\mathsf{R}_2(\psi_R,s_1,s_2,0)\\
\wedge\\
\top\gets\mathsf{R}_2(\psi_R,s_1,s_2',1)\\
\end{array}
:
\begin{array}{r}
(r_1,\psi_R)\gets\mathsf{R}_1(1^\secp)\\
(s_1,s_2,s_2')\gets\cA(r_1)
\end{array}
\right]    
\le
\negl(\secp).
\end{align}

\end{definition}

\subsection{Commitments imply dCRPuzzs}
\begin{theorem}\label{thm:com_to_dCRP}
If two-message honest-statistically-hiding computationally-binding bit commitments with classical communication exist,
then dCRPuzzs exist.
\end{theorem}

\begin{proof}
Let $(\mathsf{R}_1,\mathsf{R}_2,\mathsf{S}_1,\mathsf{S}_2)$ be
a two-message honest-statistically-hiding computationally-binding bit commitment scheme with classical communication.

Define the following algorithm $\cC$:
\begin{enumerate}
    \item 
    Get $r_1$ as input.
    \item 
    Run $(s_1,\psi_S)\gets\mathsf{S}_1(r_1,0)$.
    Run $s_2\gets\mathsf{S}_2(0,\psi_S)$.
    \item 
    Generate $\psi_S$.
    Run $s_2'\gets\mathsf{S}_2(1,\psi_S)$.
   \item 
   Output $(s_1,s_2,s_2')$.
\end{enumerate}
Let $\Pi$ be a POVM element corresponding to the event that the challenger of the security game of binding accepts.
Then, from the correctness and statistical hiding of the commitment scheme,
\begin{align}
\label{correctnessandhiding}
\sum_{r_1}\Pr[r_1\gets\mathsf{R}_1(1^\secp)]    
\Tr[\Pi(\psi_R^{\otimes 2}\otimes \cC(r_1))]\ge1-\frac{1}{q(\secp)}
\end{align}
for a certain polynomial $q$.
We will show it later.
If we define the set 
\begin{align}
V\coloneqq\left\{r_1:
\Tr[\Pi(\psi_R^{\otimes 2}\otimes \cC(r_1))]\ge1-\frac{1}{\sqrt{q(\secp)}}
\right\},
\end{align}
we have
\begin{align}
\label{correctnessandhiding}
\sum_{r_1\in V}\Pr[r_1\gets\mathsf{R}_1(1^\secp)]    
\ge 1-\frac{1}{\sqrt{q(\secp)}}
\end{align}
from the standard average argument.

From the commitment scheme, we construct a dCRPuzz $(\Setup,\Samp)$ as follows.
\begin{itemize}
    \item 
    $\Setup(1^\secp)\to\pp:$
    Run $(r_1,\psi_R)\gets\mathsf{R}_1(1^\secp)$.
    Output $\pp\coloneqq r_1$.
    \item 
    $\Samp(\pp)\to(\puzz,\ans):$
 \begin{enumerate}
 \item
    Parse $\pp=r_1$.
    \item 
    Run $(s_1,\psi_S)\gets\mathsf{S}_1(r_1,0)$.
    \item 
    Choose $b\gets\bit$.
    \item 
    Run $s_2\gets\mathsf{S}_2(b,\psi_S)$.

    \item 
    Output $\puzz\coloneqq s_1$ and $\ans\coloneqq (b,s_2)$.
\end{enumerate}

\end{itemize}
For the sake of contradiction, assume that it is not a dCRPuzz.
Then for any polynomial $p$, there exists a QPT algorithm $\cA$ such that
\begin{align}
\SD((\pp,\cA(\pp))_{\pp\gets\Samp(1^\secp)},(\pp,\mathsf{Col}(\pp))_{\pp\gets\Samp(1^\secp)})    
\le\frac{1}{p(\secp)}
\end{align}
for infinitely-many $\secp\in\mathbb{N}$.
This means that
\begin{align}
\frac{1}{p(\secp)}    
&\ge \sum_{r_1}\Pr[r_1\gets\mathsf{R}_1(1^\secp)]\TD[\cA(r_1),\mathsf{Col}(r_1)].
\end{align}
If we define the set
\begin{align}
G\coloneqq\{r_1:
\TD[\cA(r_1),\mathsf{Col}(r_1)]\le\frac{1}{\sqrt{p(\secp)}}
\},    
\end{align}
we have
\begin{align}
\sum_{r_1\in G}\Pr[r_1\gets\mathsf{R}_1(1^\secp)]
\ge1-\frac{1}{\sqrt{p(\secp)}}
\end{align}
from the standard average argument. 

From $\cA$, we construct a QPT adversary $\cB$ that breaks the binding of the commitment scheme as follows.
\begin{enumerate}
    \item 
    Get $r_1$ as input.
    \item 
    Run $(\puzz,\ans,\ans')\gets\cA(r_1)$.
    \item 
    Parse $\puzz=s_1$, $\ans=b\|s_2$, and $\ans'=b'\|s_2'$.
    \item 
    Output $(s_1,s_2,s_2')$.
\end{enumerate}
Let $\Pi$ be a POVM element corresponding to the event that the challenger of the security game of the binding accepts.
The probability that $\cB$ wins is
\begin{align}
&\sum_{r_1}\Pr[r_1\gets\mathsf{R}_1(1^\secp)]\Tr[\Pi (\psi_R^{\otimes 2}\otimes \cB(r_1))]   \\
&=\sum_{r_1}\Pr[r_1\gets\mathsf{R}_1(1^\secp)]\Tr[\Pi (\psi_R^{\otimes 2} \otimes \cA(r_1))]\\    
&\ge\sum_{r_1\in G}\Pr[r_1\gets\mathsf{R}_1(1^\secp)]\Tr[\Pi (\psi_R^{\otimes 2} \otimes \cA(r_1))]\\    
&\ge\sum_{r_1\in G}\Pr[r_1\gets\mathsf{R}_1(1^\secp)]\Tr[\Pi (\psi_R^{\otimes 2} \otimes \mathsf{Col}(r_1))]
-\frac{1}{\sqrt{p(\secp)}}\\    
&\ge\frac{1}{4}\sum_{r_1\in G}\Pr[r_1\gets\mathsf{R}_1(1^\secp)]\Tr[\Pi (\psi_R^{\otimes 2} \otimes \cC(r_1))]
-\frac{1}{\sqrt{p(\secp)}}\\    
&\ge\frac{1}{4}\sum_{r_1\in G\cap V}\Pr[r_1\gets\mathsf{R}_1(1^\secp)]\Tr[\Pi (\psi_R^{\otimes 2} \otimes \cC(r_1))]
-\frac{1}{\sqrt{p(\secp)}}\\    
&\ge \frac{1}{4}\left(1-\frac{1}{\sqrt{q(\secp)}}-\frac{1}{\sqrt{p(\secp)}}\right)\left(1-\frac{1}{\sqrt{q(\secp)}}\right)
-\frac{1}{\sqrt{p(\secp)}}\\
&\ge\frac{1}{\poly(\secp)}.
\end{align}

Let us show \cref{correctnessandhiding}.
For each $b\in\bit$, 
because of the correctness,
\begin{align}
\sum_{r_1}\Pr[r_1\gets\mathsf{R}_1(1^\secp)]\sum_{s_1}\Pr[s_1\gets\mathsf{S}_1(r_1,b)]
\sum_{s_2}\Pr[s_2\gets\mathsf{S}_2(\psi_S,b)]\Pr[\top\gets\mathsf{R}_2(\psi_R,s_1,s_2,b)]
\ge1-\negl(\secp).
\end{align}
Let $p$ be a polynomial.
For each $b\in\bit$,
if we define the set 
\begin{align}
G_b\coloneqq\left\{(r_1,s_1):
\sum_{s_2}\Pr[s_2\gets\mathsf{S}_2(\psi_S,b)]\Pr[\top\gets\mathsf{R}_2(\psi_R,s_1,s_2,b)]
\ge1-\frac{1}{p(\secp)}
\right\},    
\end{align}
we have
\begin{align}
\sum_{(r_1,s_1)\in G_b}\Pr[r_1\gets\mathsf{R}_1(1^\secp)]\Pr[s_1\gets\mathsf{S}_1(r_1,b)]
\ge 1-\negl(\secp)
\label{ave}
\end{align}
from the standard average argument.
Because of the statistical hiding,
\begin{align}
\negl(\secp)
&\ge\sum_{(r_1,s_1)}\left|
\Pr[r_1\gets\mathsf{R}_1(1^\secp)]\Pr[s_1\gets\mathsf{S}_1(r_1,0)]
-\Pr[r_1\gets\mathsf{R}_1(1^\secp)]\Pr[s_1\gets\mathsf{S}_1(r_1,1)]\right|\\
&\ge\sum_{(r_1,s_1)\in G_1}\left|
\Pr[r_1\gets\mathsf{R}_1(1^\secp)]\Pr[s_1\gets\mathsf{S}_1(r_1,0)]
-\Pr[r_1\gets\mathsf{R}_1(1^\secp)]\Pr[s_1\gets\mathsf{S}_1(r_1,1)]\right|\\
&\ge\left|
\sum_{(r_1,s_1)\in G_1}
\Pr[r_1\gets\mathsf{R}_1(1^\secp)]\Pr[s_1\gets\mathsf{S}_1(r_1,0)]
-
\sum_{(r_1,s_1)\in G_1}
\Pr[r_1\gets\mathsf{R}_1(1^\secp)]\Pr[s_1\gets\mathsf{S}_1(r_1,1)]\right|.
\end{align}
Therefore from the last inquality and \cref{ave} with $b=1$, we have
\begin{align}
\sum_{(r_1,s_1)\in G_1}\Pr[r_1\gets\mathsf{R}_1(1^\secp)]\Pr[s_1\gets\mathsf{S}_1(r_1,0)]
\ge 1-\negl(\secp).
\end{align}

Hence
\begin{align}
&\sum_{r_1}\Pr[r_1\gets\mathsf{R}_1(1^\secp)]    
\Tr[\Pi(\psi_R^{\otimes 2}\otimes \cC(r_1))]\\
&=
\sum_{(r_1,s_1)}
\Pr[r_1\gets\mathsf{R}_1(1^\secp)]\Pr[s_1\gets\mathsf{S}_1(r_1,0)]
\sum_{s_2}\Pr[s_2\gets\mathsf{S}_2(\psi_S,0)]\Pr[\top\gets\mathsf{R}_2(\psi_R,s_1,s_2,0)]\\
&\times\sum_{s_2'}\Pr[s_2'\gets\mathsf{S}_2(\psi_S,1)]\Pr[\top\gets\mathsf{R}_2(\psi_R,s_1,s_2',1)]\\
&\ge
\sum_{(r_1,s_1)\in G_0\cap G_1}
\Pr[r_1\gets\mathsf{R}_1(1^\secp)]\Pr[s_1\gets\mathsf{S}_1(r_1,0)]
\sum_{s_2}\Pr[s_2\gets\mathsf{S}_2(\psi_S,0)]\Pr[\top\gets\mathsf{R}_2(\psi_R,s_1,s_2,0)]\\
&\times\sum_{s_2'}\Pr[s_2'\gets\mathsf{S}_2(\psi_S,1)]\Pr[\top\gets\mathsf{R}_2(\psi_R,s_1,s_2',1)]\\
&\ge
(1-\negl(\secp))\left(1-\frac{1}{p(\secp)}\right)^2.
\end{align}
\end{proof}